\documentclass[amsmath,amssymb,twocolumn,10pt,aps,pre]{revtex4-1}

\usepackage{amsthm}
\usepackage{hyperref}
\usepackage{mathtools}
\usepackage{graphicx}
\usepackage[varg]{txfonts}

\newcommand{\ii}{\mathrm{i}}
\newcommand{\ff}{\mathrm{f}}
\newcommand{\e}{\mathrm{e}}
\newcommand{\g}{\mathrm{g}}
\newcommand{\tr}{\mathrm{Tr}}
\DeclareMathOperator{\dirac}{\delta_\textsc{D}}

\newcommand{\tens}[1]{\mathbf{#1}}
\newcommand{\tensgr}[1]{\boldsymbol{#1}}

\newcommand{\cvector}[1]{\left(\begin{array}{c}#1\end{array}\right)}
\renewcommand{\matrix}[2]{\left(\begin{array}{#1}#2\end{array}\right)}
\newcommand{\nexp}[1]{\mathrm{exp}\left\{ #1 \right\}} 			
\newcommand{\scalar}[2]{\left\langle #1, #2 \right\rangle}              
\newcommand{\ave}[1]{\left\langle #1 \right\rangle}			

\newcommand{\upi}[1]{\int \mathcal{D} #1 \,}				
\newcommand{\bpi}[3]{\int\limits_{#2}^{#3} \mathcal{D} #1 \,} 		
\newcommand{\usi}[1]{\int \mathrm{d} #1 \,}				
\newcommand{\bsi}[3]{\int\limits_{#2}^{#3} \mathrm{d} #1 \,}		
\newcommand{\umi}[2]{\int \mathrm{d}^{#1} #2 \,}			
\newcommand{\fmi}[2]{\int \frac{\mathrm{d}^{#1} #2}{(2\pi)^{#1}} \,}    

\newcommand{\pder}[2]{\frac{\partial #1}{\ii \partial #2}}		
\newcommand{\fder}[2]{\frac{\delta{#1}}{\ii \delta{#2}}}		
\newcommand{\mpd}{\bar{\rho}}

\newtheorem{theorem}{Theorem}
\newtheorem{corollary}[theorem]{Corollary}

\begin{document}

\title{Non-equilibrium statistical field theory for classical particles: Initially correlated grand canonical ensembles}
\author{Felix Fabis, Daniel Berg, Elena Kozlikin, Matthias Bartelmann}
\affiliation{Heidelberg University, Zentrum f\"ur Astronomie, Institut f\"ur Theoretische Astrophysik, Philosophenweg 12, 69120 Heidelberg, Germany}

\begin{abstract}
It was recently shown in \citet{Bartelmann2014} how correlated initial conditions can be introduced into the statistical field theory for classical particles pioneered by
\citet{Das2012}. In this paper we extend this development from the canonical to the grand canonical ensemble for a system satisfying statistical homogeneity
and isotropy. We do this by translating the probability distribution for the initial phase space coordinates of the particles into an easy diagrammatic representation and then
using a variant of the Mayer cluster expansion to sum over particle numbers. The grand canonical generating functional is then used in a structured approach
to the derivation of the non-interacting cumulants of the two core collective fields, the density $\rho$ and the response field $B$. As a side-product we find several
theorems pertaining to these cumulants which will be useful when investigating the interacting regime of the theory in future work. 
\end{abstract}

\maketitle

\section{Introduction}

This work is part of a series aiming at finding a new approach to the problem of non-equilibrium particle kinetics using a statistical field theory for classical particles.
While this work was begun with an application to cosmological structure formation in mind, we feel that our developments may be useful in other fields of statistical physics which
is why we will try to keep our discussions as general as possible and only make the connection with cosmology when we deem it necessary.

The statistical field theory which forms the basis for our work was developed by \citet{Mazenko2010,Mazenko2011} and \citet{Das2012,Das2013}. The basic premise of their work
was to use the path integral approach for classical mechanics (cf.~\citet{Martin1973}, \citet{Gozzi1989}, \citet{Penco2006}) in order to describe the microscopic degrees of freedom
of individual particles. Any macroscopic field is then collectively constructed from the microscopic information. This has several advantages over the standard approach of applying
the path integral formalism directly to an effective theory for the macroscopic fields which were already discussed in \citet{Bartelmann2014}. Two prominent examples are the relative
structural simplicity of the equations of motion and the fact that multi-streaming does not pose a problem because the macroscopic fields are only assembled from the microscopic
degrees of freedom at the time of interest by applying collective field operators.

While the work of \citeauthor{Das2012} on the theory was extensive they mostly concentrated on treating fluctuations around an equilibrium state with the help of
fluctuation-dissipation relations. However, the theory gives a lot of freedom in choosing the initial state of the system and is thus applicable to a very wide range of problems.
In \citet{Bartelmann2014} we thus explored how to develop the non-equilibrium statistics of a system with correlations between the initial positions and momenta of the particles.
This was done in the framework of a canonical ensemble with a fixed particle number. Two-particle interactions were implemented as a straightforward perturbation series in the 
interaction potential. We could then show in \citet{Bartelmann2014a}, \citep{Bartelmann2014b} that by describing the `free motion' partly with Zel'dovich trajectories, the non-linear
growth of the CDM power spectrum known from N-body simulations could be mimicked over a remarkable range of scales.

This work will extend the treatment of initially correlated systems with statistical homogeneity and isotropy to the grand canonical ensemble. We will see that this automatically
leads to a formulation of the theory in terms of the connected n-point functions or cumulants of the collective fields and allows for a structured approach to their calculation.
By contrast, in the canonical ensemble the n-point correlators seem to be the more natural quantities and a formulation in terms of cumulants must be obtained by hand which becomes
cumbersome for higher orders of perturbation theory.
In our next paper (\citep{Fabis2015a}) we will then use the findings of this work to extend the self-consistent perturbation theory developed by \citeauthor{Das2012} to include
initial correlations. It will allow us to obtain the linear growth of the CDM powerspectrum familiar from standard Eulerian perturbation theory directly from the theory itself
without resorting to using Zel'dovich trajectories, but rather using the actual free Hamiltonian trajectories and the unmodified Newtonian potential. This is possible due to the
fact that the grand canonical perturbation series sums up infinite classes of diagrams from the canonical one.

The outline of this paper is as follows. We begin in Sect.~2 by giving a clear definition of the grand canonical generating functional for a system satisfying
statistical homogeneity and isotropy. In Sect.~3 we implement the initial phase space probability distribution of \citet{Bartelmann2014} and express it in terms of a simple
diagrammatic language. A technique known as the Mayer cluster expansion is then used to factorize the generating functional such that the summation over particle numbers can be
performed exactly. In Sect.~4 we develop a systematic approach for managing the remaining combinatorics in deriving the non-interacting cumulants which is however reduced when
compared to the canonical approach. Along the way we derive some general theorems for these cumulants which will prove helpful for dealing with perturbation theory in
\citet{Fabis2015a}.

While many of the quantities we will encounter have already been descibed in \citep{Mazenko2010,Mazenko2011,Das2012,Das2013} we will try to stay as
close as possible to the notation introduced in \citep{Bartelmann2014} throughout this work. We refer the reader to the latter whenever we use some quantity without
giving an explicit definition.

\section{Grand canonical generating functional}

\subsection{Definition}

For the following discussion it will be helpful to think of a generating functional as the normalization factor of some probability density $\mathcal{P}$ for the $6N$ phase space
coordinates $\tens{x}$ of a collection of $N$ particles in some volume $V$. The basic concept behind the theory is to generalize the canonical ensemble from the equilibrium
Boltzmann distribution to any initial distribution $P_\ii$ and then to fix its evolution up to some arbitrary final time $t_\ff$ by requiring it to follow the classical trajectories
$\tens{x}^{\mathrm{cl}}(t;\tens{x}(t_\ii))$, i.e.~the solution to some equations of motion which is determined by choosing of an initial state $\tens{x}(t_\ii)$. This results in 
a new phase space probability density
\begin{equation}
 \mathcal{P}_{P_\ii,V,N}\left[\tens{x}(t_\ff),\tens{x}(t_\ii)\right] = P_{\ii}(\tens{x}(t_\ii)) \bpi{\tens{x}}{\ii}{\ff} \dirac\left[\tens{x}(t) - \tens{x}^{\mathrm{cl}}(t;\tens{x}(t_\ii))\right] \;.
\label{eq:02-1}
\end{equation}
The canonical generating functional is obtained by integrating over both the initial and final states. In this framework defining the grand canonical ensemble is
conceptually an easy step. We give up the notion of a fixed number of particles $N$ and replace it with an
arbitrary probability distribution for the number of particles $P_N$. Using conditional probabilities we may write
\begin{equation}
\mathcal{P}_{P_\ii,V,P_N}\left[ \tens{x}(t_\ff),\tens{x}(t_\ii),N \right] = \mathcal{P}_{P_\ii,V,P_N}\left[ \tens{x}(t_\ff), \tens{x}(t_\ii) | N \right] \, P_N(N) \;.
\label{eq:02-2}
\end{equation}
However, we know that the first factor must be the probability density in the canonical case where the particle number $N$ is fixed
\begin{equation}
\mathcal{P}_{P_\ii,V,P_N}\left[ \tens{x}(t_\ff), \tens{x}(t_\ii), N \right] = \mathcal{P}_{P_\ii,V,N}\left[\tens{x}(t_\ff),\tens{x}(t_\ii)\right] \, P_N(N) \;.
\label{eq:02-3}
\end{equation}
With this we may define the grand canonical partition functional simply as
\begin{equation}
Z_{\mathrm{gc}} = \sum_{N=0}^{\infty} \, \upi{\tens{x}(t)} P_\ii(\tens{x}(t_\ii)) \, P_N(N) \, \dirac\left[\tens{x}(t) - \tens{x}^{\mathrm{cl}}(t;\tens{x}(t_\ii))\right] \;,
\label{eq:02-4}
\end{equation}
where we absorbed the integration over initial and final states into the path integral.

\subsection{Particle number probability distribution}

We now need to specify the probability density $P_N$ for the number of particles. In order to do so we use the familiar textbook approach of embedding our system
$\mathcal{S}_{\mathrm{gc}}$ into a much larger canonical system $\mathcal{S}_\mathrm{c}$.
The grand canonical system $\mathcal{S}_{\mathrm{gc}}$ may exchange particles with its complement in $\mathcal{S}_{\mathrm{c}}$ and particles may interact across
the boundary enclosing $\mathcal{S}_{\mathrm{gc}}$.

\begin{figure}[htp]
 \centering
 \includegraphics[scale=0.5]{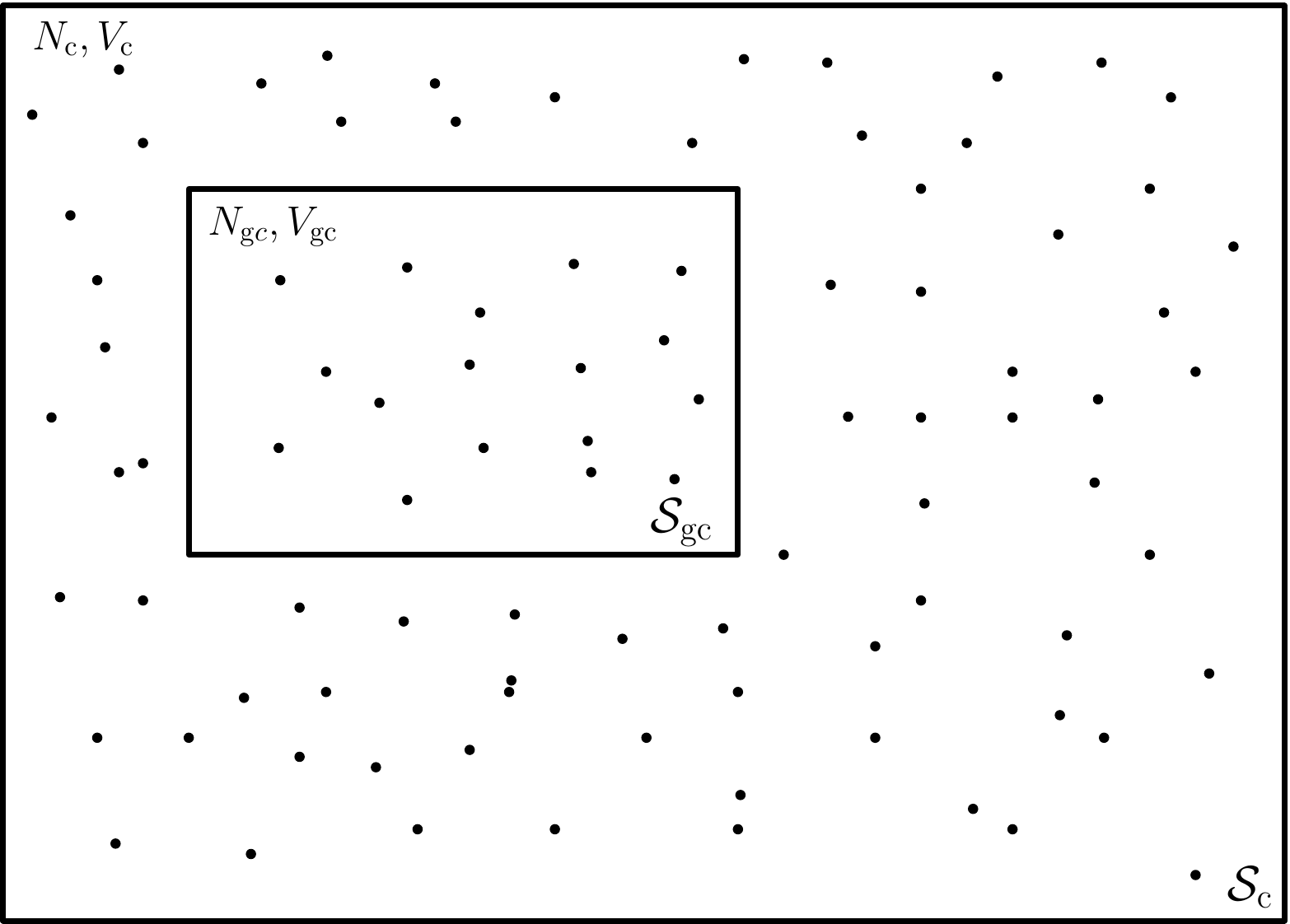}
 \caption{Embedding of the grand canonical system into a much bigger canonical system.}
\end{figure}

The standard approach for obtaining $P_N$ in equilibrium statistical physics would be to introduce a chemical potential defining the amount of energy needed for a single 
particle exchange between the two systems. Since we do not require equilibrium it would be in general quite hard to define this quantity since it
might in principle depend on the entire instanteneous phase space configuration $\tens{x}(t)$. Out of equilibrium we also lack a temperature defining an energy scale with
which we can compare the chemical potential.

We can circumvent this problem by restricting ourselves to the case where the canonical system $\mathcal{S}_\mathrm{c}$ is statistically homogeneous and isotropic, i.e.~its
statistical properties on all scales of interest are invariant under translation and rotation at all times. This then means that the probability for finding an individual particle
somewhere inside $\mathcal{S}_\mathrm{c}$ must be equal everywhere regardless of any kind of interactions or correlations. From this we can then immediately conclude that the
probability $p$ for finding a particle inside the subsystem $\mathcal{S}_\mathrm{gc}$ is given by
\begin{equation}
p = \frac{V_{\mathrm{gc}}}{V_\mathrm{c}} \;.
\label{eq:02-5}
\end{equation}
The probability for $N_{\mathrm{gc}}$ particles to be in $\mathcal{S}_{\mathrm{gc}}$ is thus given by a binomial distribution
\begin{equation}
N_{\mathrm{gc}} \sim \mathcal{B}\left(N_\mathrm{c},p=\frac{V_{\mathrm{gc}}}{V_\mathrm{c}}\right) \;.
\label{eq:02-6}
\end{equation}
We now take the `thermodynamic' limit of $N_\mathrm{c},V_\mathrm{c} \rightarrow \infty$ while keeping the mean particle density $\mpd = N_\mathrm{c}/V_\mathrm{c}$ and thus the
mean mass density constant. In this limit $p \rightarrow 0$ and we may approximate (\ref{eq:02-6}) by a Poisson distribution as
\begin{equation}
N_{\mathrm{gc}} \sim \mathcal{B}\left(N_\mathrm{c},p=\frac{V_{\mathrm{gc}}}{V_\mathrm{c}}\right) \; \rightarrow \; \mathcal{P}\left(N_\mathrm{c} p = \mpd \, V_{\mathrm{gc}} = \ave{N_{\mathrm{gc}}} \right) \;.
\label{eq:02-7}
\end{equation}
We now ignore the surrounding canonical system and drop the suffix `gc'. Pending normalization we may thus state
\begin{equation}
P_N(N) = \frac{\ave{N}^N}{N!} = \frac{\mpd^N \, V^N}{N!} \;.
\label{eq:02-8}
\end{equation}

\section{Initial correlations}

\subsection{Preliminaries}

We now insert \eqref{eq:02-8} into \eqref{eq:02-4} and perform all the steps described in \citet{Bartelmann2014}. We rewrite the functional Dirac delta distribution as a
Fourier transform and separate the theory into a free and an interacting part. We then introduce collective fields and their respective source fields $H$, which allows us to express the
interaction by an operator $\hat{\mathrm{S}}_\mathrm{I}$, which is independent of particle number. Then we also express the collective fields by an operator $\hat{\Phi}$, which
enables us to execute the path integrals using the solution to the free equations of motion. The complete generating functional then reads
\begin{widetext}
 \begin{align}
 Z_{\mathrm{gc}}[H,\tens{J},\tens{K}] = \e^{\ii \, \mathrm{\hat{S}}_\mathrm{I}} \, \sum_{N=0}^{\infty} \, \e^{\ii \, H \, \cdot \, \hat{\Phi}} \usi{\tens{q}^{(\ii)}} \usi{\tens{p}^{(\ii)}} \,& 
                                        \frac{\mpd^N \, V^N}{N!} \frac{V^{-N}}{\sqrt{(2\pi)^{3N}\det C_{pp}}} \, \left(\hat{\mathcal{C}} \left( \pder{}{\tens{p}^{(\ii)}} \right)
					\nexp{-\frac{1}{2} \tens{p}^{(\ii)\top} \, C_{pp}^{-1} \, \tens{p}^{(\ii)}} \right) \notag \\
				       &\e^{\ii \, \scalar{\bar{\tens{J}}_q}{\tens{q}^{(\ii)}} } \, \e^{\ii \, \scalar{\bar{\tens{J}}_p}{\tens{p}^{(\ii)}} } \, \e^{-\ii \, S_K^{(N)}[\tens{J},\tens{K}]} \;.
\label{eq:03-1}
\end{align}
\end{widetext}
Since the interaction operator can be pulled out in front of the entire expression we only need to concern ourselves with the free grand canonical generating functional for the
remainder of this work. We directly see that the power of volume the $V$ cancel nicely. The quantity $\hat{\mathcal{C}}$ is a
\emph{polynomial} operator containing density-density correlations and density-momentum correlations. Its explicit form will be given later in \eqref{eq:03-15}.
Throughout the paper we will make implicit use of the fact that both these types of correlations as well as $C_{pp}$ only depend on the initial positions of particles, but not on
their momenta, which is due to the initial momenta being fixed to the value of an initial random momentum field. The time averaged sources are defined as 
\begin{equation}
 \bar{\tens{J}}_{q,p} = \bsi{t}{t_\ii}{t_\ff} \tens{J}(t)^\top \mathcal{G}(t,t_\ii) \, \mathcal{P}_{q,p} 
\label{eq:03-2}
\end{equation}
with the projection operators
\begin{equation}
 \mathcal{P}_q = \matrix{cc}{\mathcal{I}_3 \\ 0_3 } \otimes \mathcal{I}_N \;, \quad \mathcal{P}_q = \matrix{cc}{0_3 \\ \mathcal{I}_3} \otimes \mathcal{I}_N \;,
\label{eq:03-3}
\end{equation}
which take care of selecting either the position or momentum part from the right hand side of the free $N$-particle propagator $\mathcal{G} = G \otimes \mathcal{I}_N$, where
$G$ is the single particle free propagator, i.e.~the Green's function of the free equations of motion of a single particle. The purely source dependent action term $S_K^{(N)}$
is given by
\begin{equation}
 S_K^{(N)}[\tens{J},\tens{K}] = \bsi{t}{t_\ii}{t_\ff} \bsi{t'}{t_\ii}{t_\ff} \tens{J}(t) \, \mathcal{G}(t,t') \, \tens{K}(t) \;.
\label{eq:03-4}
\end{equation}
The angular brackets in \eqref{eq:03-1} are not averages but define the scalar product for the tensors bundling properties of all $N$ particles such that 
\begin{equation}
 \scalar{\tens{A}}{\tens{B}} = \sum_{j=1}^N \, \vec{A}_j \cdot \vec{B}_j \;.
\label{eq:03-5}
\end{equation}
We may now apply the appropriate number of partial integrations to change all partial derivatives w.r.t.~initial momenta present in the $\hat{\mathcal{C}}$ operator from acting
on the Gaussian exponential to act on the phase factor $\e^{\ii \, \scalar{\bar{\tens{J}}_p}{\tens{p}^{(i)}} }$ instead, picking up a minus sign for every partial integration.
The fact that $\hat{\mathcal{C}}$ is polynomial and $C_{pp}$ is positive-definite ensures that all boundary terms vanish. Then we execute all these derivatives to obtain
\begin{gather}
    \usi{\tens{p}^{(\ii)}} \left( \hat{\mathcal{C}} \left( \pder{}{\tens{p}^{(\ii)}} \right) \, \nexp{-\frac{1}{2} \tens{p}^{(\ii)\top} \, C_{pp}^{-1} \, \tens{p}^{(\ii)}} \right) \, \e^{\ii \, \scalar{\bar{\tens{J}}_p}{\tens{p}^{(\ii)}} } \notag \\
 =  \usi{\tens{p}^{(\ii)}} \nexp{-\frac{1}{2} \tens{p}^{(\ii)\top} \, C_{pp}^{-1} \, \tens{p}^{(\ii)}} \, \hat{\mathcal{C}} \left(-\bar{\tens{J}}_p \right) \e^{\ii \, \scalar{\bar{\tens{J}}_p}{\tens{p}^{(\ii)}} } \;.
\label{eq:03-6}
\end{gather}
The remaining integration over initial momenta is now a $3N$-dimensional Fourier transform from $\tens{p}^{(\ii)}$ to $\bar{\tens{J}}_p$ which gives
\begin{gather}
    \usi{\tens{p}^{(\ii)}} \frac{1}{\sqrt{(2\pi)^{3N}\det C_{pp}}} \nexp{-\frac{1}{2} \tens{p}^{(\ii)\top} \, C_{pp}^{-1} \, \tens{p}^{(\ii)}} \, \e^{\ii \, \scalar{\bar{\tens{J}}_p}{\tens{p}^{(\ii)}} } \notag \\
 =  \,\nexp{-\frac{1}{2} \bar{\tens{J}}_p^\top \, C_{pp} \, \bar{\tens{J}} } \;.
\label{eq:03-7}
\end{gather}
In this form we can easily split up $C_{pp}$ into its diagonal part containing the auto-correlations of the momenta of individual particles and the remaining trace-free part which
contains only cross-correlations between momenta of \emph{different} particles and has entries of $3\times3$ dimensional zero matrices on the diagonal.
The statistical homogeneity of our system dictates that all entries on the diagonal must be independent of particle position and thus spatially constant. We may thus split $C_{pp}$
as
\begin{equation}
 C_{pp} = \left( \sigma_p^2 \, \mathcal{I}_3 \right) \otimes \mathcal{I}_N + \sum_{j \neq k} \ave{\vec{p}_{j}^{\,(\ii)} \otimes \vec{p}_{k}^{\,(\ii)}} \otimes E_{jk} \;,
\label{eq:03-8}
\end{equation}
where $E_{jk} = \vec{e}_j \otimes \vec{e}_k$. When we later specialize to a curl-free initial velocity field we find  $\sigma_p^2 = \alpha \sigma_1^2/3$ where $\alpha$
is the constant conversion factor between velocity and momentum. For later use we define the second term as the cross-correlation matrix 
\begin{equation}
 C_{pp}^* \coloneqq \sum_{j \neq k} \ave{\vec{p}_{j}^{\,(\ii)} \otimes \vec{p}_{k}^{\,(\ii)}} \otimes E_{jk} = \sum_{j \neq k} C_{p_j p_k} \otimes E_{jk} \;.
\label{eq:03-9}
\end{equation}
Observe that in the notation of \citet{Bartelmann2014} $C_{p_j p_k} = \left(B_{pp}\right)_{jk}$. The diagonal part can be used to reintroduce
the integral over initial momenta by reversing the Fourier transform. This leads to
\begin{gather}
   \nexp{-\frac{1}{2} \, \bar{\tens{J}}_p^\top \, (\sigma_p^2 \, \mathcal{I}_{3N}) \, \bar{\tens{J}}_p} = \usi{\tens{p}^{(\ii)}} \frac{\nexp{-\frac{\scalar{\tens{p}^{(\ii)}}{\tens{p}^{(\ii)}}}{2 \, \sigma_p^2}}}{\sqrt{(2\pi\sigma_p^2)^{3N}}} \, \e^{\ii \, \scalar{\bar{\tens{J}}_p}{\tens{p}^{(\ii)}} } \notag \\
 = \usi{\tens{p}^{(\ii)}} P^{\mathrm{MB}}_{\sigma_p}\left(\tens{p}^{(\ii)}\right) \, \e^{\ii \, \scalar{\bar{\tens{J}}_p}{\tens{p}^{(\ii)}} } \;,
\label{eq:03-10}
\end{gather}
where the $P^{\mathrm{MB}}_{\sigma_p}$ is now the Maxwell-Boltzmann distribution with momentum dispersion $\sigma_p$. This is of course due to the fact that we have chosen a
Gaussian random velocity field in the first place and just separated off the cross-correlations. While reintroducing an already performed integration may seem as a step backwards
it will have the very desirable effect that we again have the complete free solution of the path integrals in the generating functional such that functional
derivatives w.r.t.~the sources $\tens{J},\tens{K}$ can be replaced by the phase-space quantities $\tens{x},\tensgr{\chi}$ conjugate to these sources. This will be important
for performing the Mayer cluster expansion. Finally, we note that we may replace the time averaged
source $\bar{\tens{J}}_p$ that appear in both $\hat{\mathcal{C}}$ and the momentum cross-correlation Gaussian as a functional derivative by using the relation
\begin{equation}
   - \bar{\tens{J}}_p  \, \e^{-\ii \, S_K^{(N)}[\tens{J},\tens{K}]} = \fder{}{\tens{K}_p(t_\mathrm{i})} \, \e^{-\ii \, S_K[\tens{J},\tens{K}]} \;.
\label{eq:03-11}
\end{equation}
After these manipulations the free grand canonical generating functional has the form shown in \eqref{eq:03-12}. Interestingly this expression for the initially
correlated set of particles shows that we can obtain its generating functional by first finding the functional of an ideal gas where particles are initially uncorrelated
in configuration space and have a Maxwell-Boltzmann distribution in momentum space. The initial correlations can then be induced by applying suitable operators.
\begin{widetext}
 \begin{align}
 Z_\mathrm{gc,0}[H,\tens{J},\tens{K}] = \sum_{N=0}^\infty \, \e^{\ii \, H \cdot \hat{\Phi}} \, \hat{\mathcal{C}}\left(\fder{}{\tens{K}_p(t_\mathrm{i})}\right) &\, \nexp{-\frac{1}{2} \left( \fder{}{\tens{K}_p(t_\mathrm{i})} \right)^\top C_{pp}^* \, \left(\fder{}{\tens{K}_p(t_\mathrm{i})} \right) } \notag \\
					& \frac{\mpd^N}{N!} \usi{\tens{q}^{(\ii)}} \usi{\tens{p}^{(\ii)}} P_{\sigma_p}^{\mathrm{MB}}\left(\tens{p}^{(\ii)}\right) \e^{\ii \, \scalar{\bar{\tens{J}}_q}{\tens{q}^{(\ii)}} } \, \e^{\ii \, \scalar{\bar{\tens{J}}_p}{\tens{p}^{(\ii)}} } \, \e^{-\ii \, S_K^{(N)}[\tens{J},\tens{K}]}
\label{eq:03-12}
\end{align}
\end{widetext}

\subsection{Diagrammatic representation of the initial correlations}

In the above form \eqref{eq:03-12} the sum over particle numbers can not be performed straightforwardly  because the initial correlations prevent the factorization
of the free generating functional into single particle contributions. We thus reorganise the partition sum, representing the correlations in diagrammatic form.
We begin with the momentum cross-correlations and rewrite the Gaussian factor in \eqref{eq:03-12} in the following way, which may be recognized as
the first step of a Mayer cluster expansion:
\begin{align}
    & \nexp{-\frac{1}{2} \left( \fder{}{\tens{K}_p(t_\mathrm{i})} \right)^\top \, C_{pp}^* \, \left( \fder{}{\tens{K}_p(t_\mathrm{i})} \right) } \notag \\
   =& \, \nexp{-\frac{1}{2} \sum_{j=1}^N \, \sum_{k=1, k \neq j}^N \left( \fder{}{\vec{K}_{p_j}(t_\mathrm{i})} \right)^\top \, C_{p_j p_k} \, \left( \fder{}{\vec{K}_{p_k}(t_\mathrm{i})} \right) } \notag \\
   =& \, \prod_{\{j,k\}} \left( 1 + \left[ \nexp{- \left( \fder{}{\vec{K}_{p_j}(t_\mathrm{i})} \right)^\top \, C_{p_j p_k} \, \left( \fder{}{\vec{K}_{p_k}(t_\mathrm{i})} \right) } - 1\right] \right) \notag \\
   =:& \, \prod_{\{j,k\}} \left( 1 + \hat{C}_{p_j p_k} \right) \;.
\label{eq:03-13}
\end{align}
In the second line we expanded the quadratic form explicitly into a sum over all the different particles using the definition of $C_{pp}^*$ from \eqref{eq:03-9}. In the third line,
we used the symmetry $C_{p_j p_k} = C_{p_k p_j}$ to express the double sum as a sum over all \emph{different pairs} $\{j,k\}$. Each pair only appears once in the sum,
e.g. $\{1,2\}$ and $\{2,1\}$ are considered equivalent and only one of them is summed over. In the fourth line, we defined a new scalar operator $\hat{C}_{p_j p_k}$ describing
the momentum correlation between two particles. In Fig.~\ref{fig:03-1}, we now represent this new operator as a dashed line connecting two dots which represent particles $i$ and $j$.
\begin{figure}[htp]
 \centering
 \includegraphics[scale=0.7]{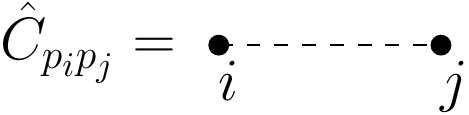}
 \caption{The $\hat{C}_{p_i p_j}$ line representing correlation between the initial momenta of two particles.}
\label{fig:03-1}
\end{figure}
We write out the product \eqref{eq:03-13} in terms of sums over \emph{different} $n$-tupels of different particle pairs as
\begin{equation}
 \prod_{\{i,j\}} \left( 1 + \hat{C}_{p_i p_j} \right) = 1 + \sum_{\{i,j\}} \hat{C}_{p_i p_j} + \sum_{\{\{i,j\},\{k,l\}\}^\star} \hat{C}_{p_i p_j} \hat{C}_{p_k p_l} + \ldots
\label{eq:03-14}
\end{equation}
The tupels are different in the same sense as the particle pairs, for example $\{\{1,2\},\{3,4\}\}$ and $\{\{3,4\},\{2,1\}\}$ are considered equivalent, while $\{\{1,3\},\{2,4\}\}$
is different. The star indicates that the pairs $\{i,j\}$ and $\{k,l\}$ may not be the same. This means that $\{\{1,2\},\{2,1\}\}$ is excluded, while $\{\{1,2\},\{2,3\}\}$
is included. This scheme extends to all $n$-tupels, leading to the only restriction on the topology of diagrams:
\begin{description}
 \item[Rule 1] Any pair of particles $i$ and $j$ may only be connected by at most one $\hat{C}_{p_i p_j}$-line.
\end{description}
We can now easily express the above sums in a diagrammatic form by going through all particle numbers $1 < n < N$ and drawing for every $n$ all diagrams compatible with
the above rule.
\begin{figure}[htp]
 \centering
 \includegraphics[scale=0.5]{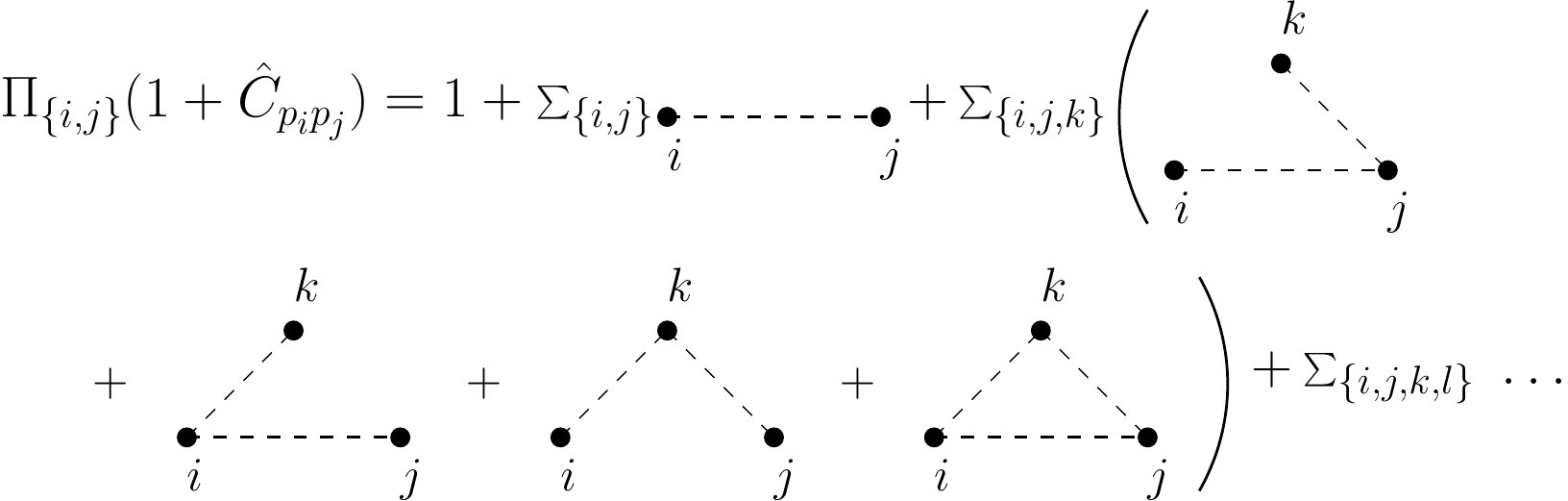}
 \caption{The momentum correlation of all particles ordered by number of correlated particles.}
\label{fig:03-2}
\end{figure}
As above, $\{i,j,k\}$ indicates that all distinct triplets are summed over, e.g.~$\{2,3,1\}$ is equivalent to $\{1,2,3\}$ and not counted extra. The same holds for all other
permutations of $\{1,2,3\}$. Because of this we had to draw three diagrams in Fig. \ref{fig:03-2} that are topologically identical which appears to be cumbersome.
However, this has the very desirable effect of making the above expression invariant under particle exchange. This property will become important later on.

The correlation operator $\hat{\mathcal{C}}$ was derived in the appendix of \citet{Bartelmann2014}. We first need to expand it explicitly with all the contributions from
the individual particles appearing.
\begin{align}
    \hat{\mathcal{C}}
 &= \prod_{l=1}^N \left( 1 - \ii \sum_{\substack{ n=1 \\ n \neq l}} \vec{C}_{\delta_l p_n} \cdot \fder{}{\vec{K}_{p_n}(t_\mathrm{i})} \right) \notag \\
 &+ \sum_{\{i,j\}} C_{\delta_i \delta_j} \, \prod_{\{l\}'} \left(1 - \ii \sum_{\substack{ n=1 \\ n \neq l}} \vec{C}_{\delta_l p_n} \cdot \fder{}{\vec{K}_{p_n}(t_\mathrm{i})} \right) \notag \\
 &+ \sum_{\{\{i,j\},\{k,l\}\}'} C_{\delta_i \delta_j} \, C_{\delta_k \delta_l} \, \prod_{\{m\}'} \left(1 - \ii \sum_{\substack{ n=1 \\ n \neq m}} \vec{C}_{\delta_m p_n} \cdot \fder{}{\vec{K}_{p_n}(t_\mathrm{i})} \right) \notag \\
 &+ \sum_{\{\{i,j\},\{k,l\},\{m,n\}\}'} \ldots
\label{eq:03-15}
\end{align}
Keep in mind that while $C_{p_i p_j}$ is a matrix, $C_{\delta_i \delta_j}$ is a scalar quantity. The restriction $n \neq l$ on the inner sums follows from
$\vec{C}_{\delta_l p_l} = 0$, which is a consequence of the statistical homogeneity and isotropy of the system. Primes on the sums mean that any particle index may appear only
in one pair of the $n$-tupel of pairs. For example, $\{\{1,2\},\{2,3\}\}$ is forbidden. Primes on the product index mean that all particle indices present in the
term of the preceding sum are excluded. Now define the new operator
\begin{equation}
 \hat{C}_{\delta_i p_j} = - \ii \, \vec{C}_{\delta_i p_j} \cdot \fder{}{\vec{K}_{p_j}(t_\mathrm{i})} \;.
\label{eq:03-16}
\end{equation}
For this operator and $C_{\delta_i \delta_j}$ we introduce line diagram representations as basic building blocks in Fig. \ref{fig:03-3}.
\begin{figure}[htp]
 \centering
 \includegraphics[scale=0.7]{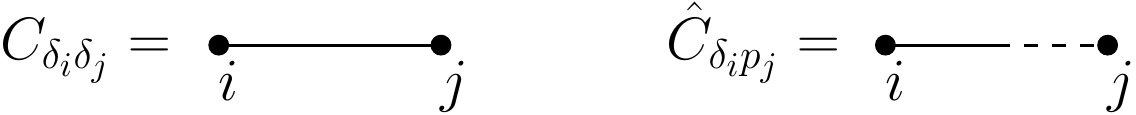}
 \caption{Representations of initial density and density-momentum cross correlation.}
 \label{fig:03-3}
\end{figure}
We again expand the products into multiple sums and then reorder the entire expression in terms of the number of particles being correlated.
This means that we can express $\hat{\mathcal{C}}$ by going through all particle numbers $\ell$ up to $N$ and drawing for each $\ell$ all diagrams that conform with a set of rules
that we read off \eqref{eq:03-15} and then summing these diagrams over all $\ell$-tupels of particles.
\begin{description}
 \item[Rule 2] Any particle index appears only once in the terms of the sums $\sum_{\{\{i_1,j_1\},\ldots,\{i_n,j_n\}\}'}$ over $\delta$-indices in \eqref{eq:03-15}.
               Thus, no $C_{\delta_i \delta_j}$ may be connected to one another. 
 \item[Rule 3] Since each product index in \eqref{eq:03-15} may appear only once in every term, no $\hat{C}_{\delta_l p_n}$ may be connected to another with the solid $\delta$ end.
 \item[Rule 4] The products in \eqref{eq:03-15} have no common index with the preceding sums due to the $\{l\}'$ restriction. Thus, no $\hat{C}_{\delta_l p_n}$ may be connected
               with its solid $\delta_l$ end to a $C_{\delta_l \delta_j}$.
\end{description}
We provide an example for each of the forbidden topologies in Fig.~\ref{fig:03-}. They can summarily be expressed as: `No solid $\delta$ lines may meet at the same particle'.
\begin{figure}[htp]
 \centering
 \includegraphics[scale=0.45]{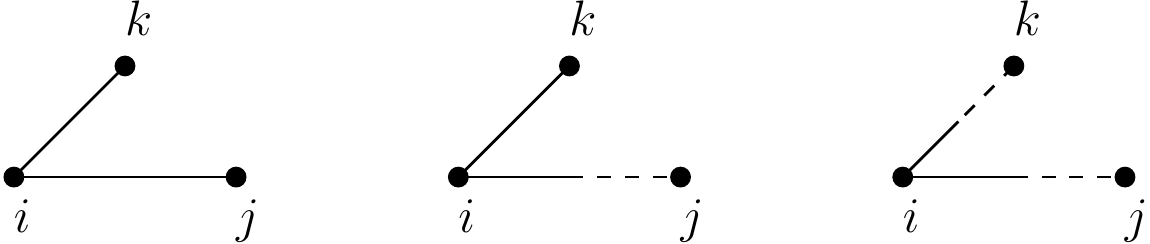}
 \caption{All three of the forbidden topologies according to rules 2-4.}
\label{fig:03-}
\end{figure}
\begin{figure}[htp]
 \centering
 \includegraphics[scale=0.45]{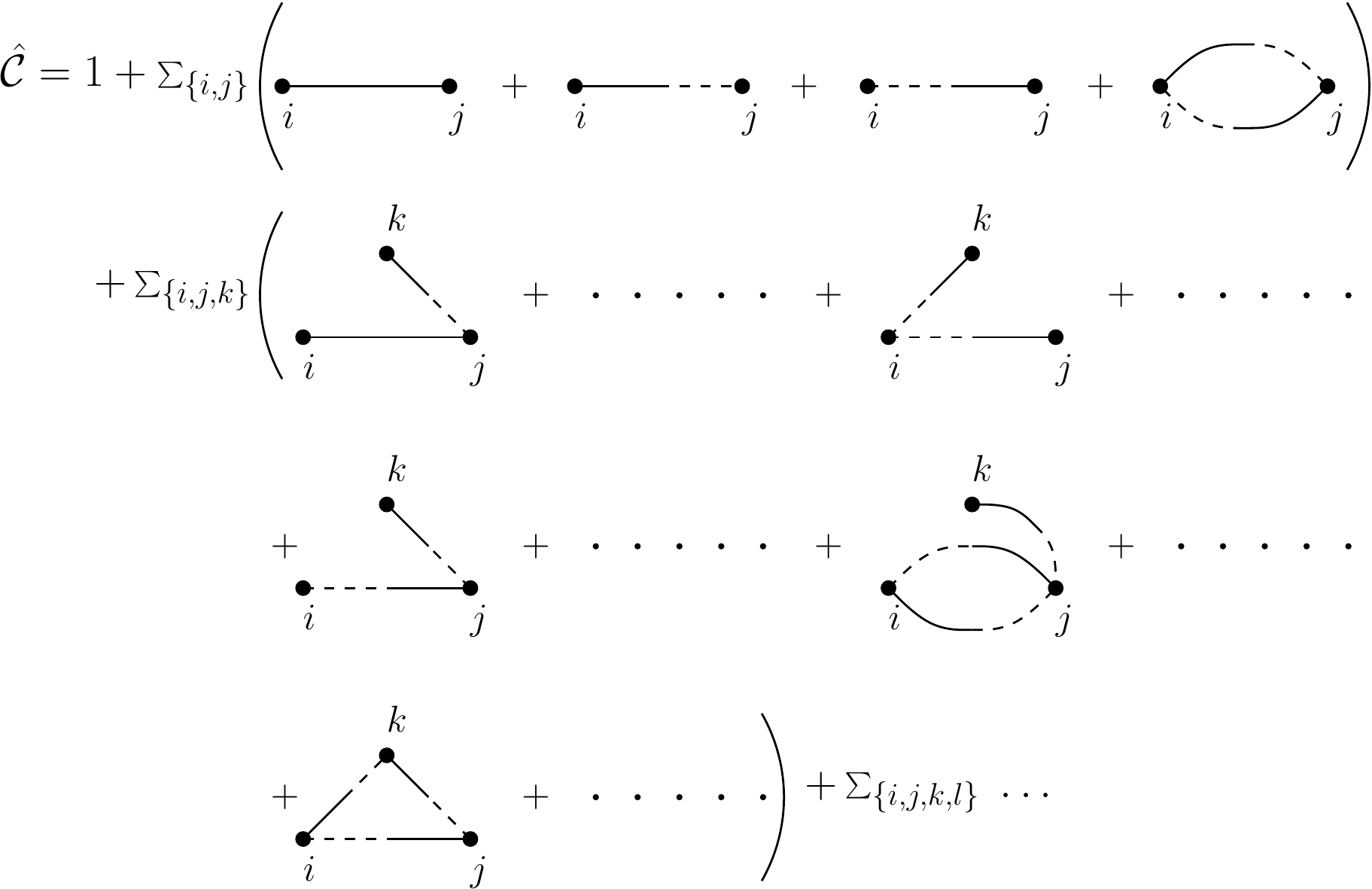}
 \caption{Diagrammatic expansion of (\ref{eq:03-15}) ordered by the number of correlated particles.}
\label{fig:03-4}
\end{figure}
Fig.~\ref{fig:03-4} shows the resulting diagrammatic expansion of the correlation operator $\hat{\mathcal{C}}$. The dots behind the diagrams correlating three particles
stand for all other diagrams which are topologically equivalent to the preceding one. According to \eqref{eq:03-12} we now need the product of both expansions shown in Figures
\ref{fig:03-2} and \ref{fig:03-4}. The principal form of our new total correlation operator is thus
\begin{equation}
 \hat{\mathcal{C}}_{\mathrm{tot}} = (1 + \Sigma_1) \, (1 + \Sigma_2) = 1 + \Sigma_1 + \Sigma_2 + \Sigma_1 \, \Sigma_2 \;,
\label{eq:03-17}
\end{equation}
which means that we retain all diagrams of both original expansions and must add all possible diagrams that can be drawn by using all three line types adhering to their individual rules,
resulting in the expansion shown in Fig. \ref{fig:03-5}. We have only drawn diagrams with two particles here which must of course be connected.
However, as soon as we arrive at four particles disconnected diagrams appear, the simplest example being two solid lines representing
for example $C_{\delta_1 \delta_2} C_{\delta_3 \delta_4}$.
\begin{figure}[htp]
 \centering
 \includegraphics[scale=0.45]{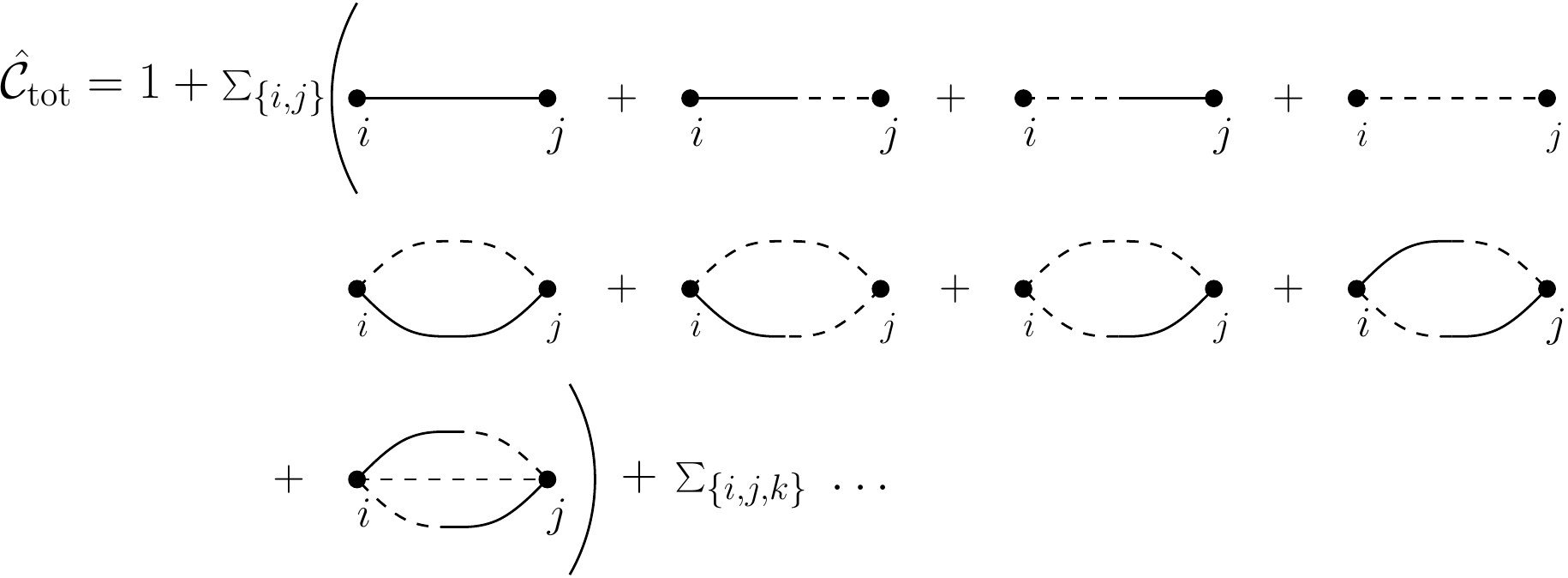}
 \caption{Diagrammatic expansion of the total correlation operator up to two particles.}
 \label{fig:03-5}
\end{figure}

\subsection{Mayer cluster expansion}

The Mayer cluster expansion is a technique first introduced in \citet{Mayer1941}. Since we already have introduced a diagrammatic form for our generating functional we will use
a suitable formulation found in \citet{becker}, adopting most of its notation and reasoning. The cluster expansion is normally used when one wants to compute the partition sum of
an interacting gas in equilibrium statistical physics. The configuration integral cannot be performed exactly due to the interactions. The Mayer cluster expansion is then employed
to expand the configuration integral into a perturbative series ordered by the number of particles taking part in the interaction allowing to execute the spatial integral
at least numerically. A common application is then to derive expressions for the coefficients in the virial expansion of the equation of state, which corrects the non-interacting
equation of state by terms of higher than linear order in the mean density.

As mentioned before, in our case the interaction has been separated into an operator and can be ignored for now. The Mayer cluster expansion is however more general in nature and
may be applied to our problem of managing the initial correlations. As a first step we remind ourselves of the following result derived Eq.~(35) of \citet{Bartelmann2014} for
the free $N$-particle generating functional, where we choose to leave out the averaging over the initial phase-space state $\tens{x}^{(\ii)}$:
\begin{align}
 Z_0^{(N)}[\tens{J},\tens{K}] &:= \bpi{\tens{x}(t)}{\mathrm{i}}{} \upi{\tensgr{\chi}(t)} \e^{\ii \, \mathrm{S}_0 + \ii \bsi{t}{t_\mathrm{i}}{t_\mathrm{f}} \left[ \scalar{\tens{J}(t)}{\tens{x}(t)} + \scalar{\tens{K}(t)}{\tensgr{\chi}(t)} \right] } \notag \\
                              &= \e^{\ii \, \scalar{\bar{\tens{J}}_q}{\tens{q}^{(\ii)}} } \, \e^{\ii \, \scalar{\bar{\tens{J}}_p}{\tens{p}^{(\ii)}} } \, \e^{-\ii \, S_K^{(N)}[\tens{J},\tens{K}]} \;,
\label{eq:03-18}
\end{align}
where $\mathrm{S}_0$ is the action containing only the free equations of motion. 
Inserting the first line into \eqref{eq:03-12}, we may pull both the collective field operator $\hat{\Phi}$ and our total correlation operator
$\hat{\mathcal{C}}_{\mathrm{tot}}$ back under the path integrals and execute the functional derivatives contained in them. This replaces the collective field operator $\hat{\Phi}$ 
with the actual collective field $\Phi$ and the argument of $\hat{\mathcal{C}}_{\mathrm{tot}}$ with
\begin{equation}
 \fder{}{\tens{K}_{p}(t_\ii)} \rightarrow \tensgr{\chi}_p(t_\ii) \;.
\label{eq:03-19}
\end{equation}
For now, we set the source terms $\tens{J},\tens{K}$ to zero because they are irrelevant for the following steps. They will be reintroduced later.
With the exception of $\hat{\mathcal{C}}_{\mathrm{tot}}$ all other parts of the generating functional can now be factorized into contributions from single particles.
For the differentials in the integral over the initial state and in the path integrals this is trivial and for the Maxwell-Boltzmann distribution we have
\begin{equation}
 P^{\mathrm{MB}}_{\sigma_p}(\tens{p}^{(\ii)}) = \prod_{j=1}^N \left( \frac{\e^{-\frac{(p_j^{(\ii)})^2}{2 \sigma_p^2}} }{\sqrt{(2\pi \, \sigma_p^2)^3}} \right) = \prod_{j=1}^N P^{\mathrm{MB}}_{\sigma_p} (\vec{p}_j^{\,(\ii)}) \;.
\label{eq:03-20}
\end{equation}
By definition (cf. \citet{Bartelmann2014}) the collective field $\Phi$ can be separated into one particle contributions. This must also hold for the free action,
\begin{equation}
 \mathrm{S}_0 = \sum_{j=1}^N \, \mathrm{S}_{0,j} \;, \quad  \Phi(1) = \sum_{j=1}^N \phi_j(1) \;.
\label{eq:03-21}
\end{equation}
We can now gather all one-particle contributions into a trace operator defined as
\begin{align}
 \tr_j := \mpd \umi{3}{q_j^{(\ii)}} \umi{3}{p_j^{(\ii)}} P^{\mathrm{MB}}_{\sigma_p}(\vec{p}_j^{\,(\ii)}) \notag \\
          \bpi{\vec{x}_j}{\ii}{} \upi{\vec{\chi}_j} \e^{\ii \, \left( \mathrm{S}_{0,j} + H \cdot \phi_j\right)} 
\label{eq:03-22}
\end{align}
and subsequently use this to write the non-interacting grand canonical generating functional in the very compact form
\begin{equation}
 Z_{\mathrm{gc},0}[H] = \sum_{N=0}^{\infty} \, \frac{1}{N!} \, \left( \prod_{j=1}^N \tr_j \right) \, \mathcal{C}^{(N)}_{\mathrm{tot}}\left(\tensgr{\chi}_p(t_\mathrm{i})\right) \;.
\label{eq:03-23}
\end{equation}

Now consider a set of $N$ particles. We want to separate the particles into a collection $\{c_\ell\}$ of subsets containing $\ell$ particles each. One such subset $c_\ell$ is
called a `cluster of size $\ell$'. Each such collection $\{c_\ell\}$ also has a set of numbers $\{m_\ell\}$ with $m_\ell$
being the number of clusters of size $\ell$ in $\{c_\ell\}$. A set $\{m_\ell\}$ is called a `cluster configuration', where we only care about how many clusters of size $\ell$ there
are and not about which actual particles are in which clusters. Clearly, every cluster configuration $\{m_\ell\}$ has different `realisations' $\{c_\ell\}$.
Any configuration must of course obey the constraint
\begin{equation}
 \sum_{\ell=0}^N \, m_\ell \cdot \ell = N \;.
\label{eq:03-24}
\end{equation}
In our diagrammatic language such a configuration corresponds to a \emph{clustering pattern} of particle dots. This means that we group dots into clusters by drawing the dots
with a clear separation between clusters. A realisation $\{c_\ell\}$ of such a configuration is then given by specifying which actual particle is assigned to which dot.
As an example we pick a set of $N=6$ particles and the realisation $(1,2,3)(4,5)(6)$ of the cluster configuration
$\{m_1 = 1, m_2 = 1, m_3 = 1, m_4 = m_5 = m_6 = 0\}$. This realisation corresponds to the \emph{fixed} clustering pattern shown in Fig.~\ref{fig:03-6}.
The notion of `fixed' means, that one needs to map the particle indices of the sequence $(1,2,3)(4,5)(6)$ to the dots by some bijective mapping. Once this is done,
the pattern may \emph{not} be changed in any way since this would correspond to another clustering realisation like e.g.~$(2,1,3)(4,5)(6)$, where
particles $1$ and $2$ have switched places. 
\begin{figure}[htp]
 \centering 
 \includegraphics[scale=0.5]{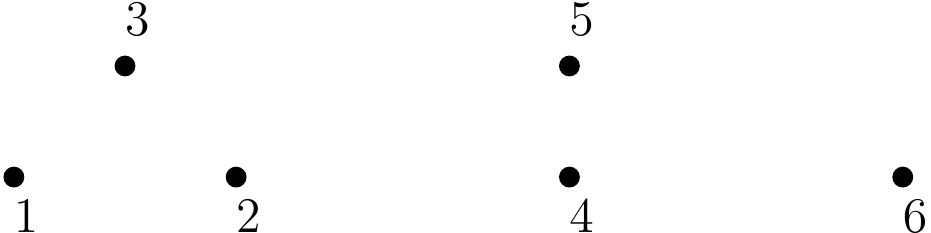}
 \caption{A fixed clustering pattern representing the $(1,2,3)(4,5)(6)$ clustering realisation.}
 \label{fig:03-6}
\end{figure}
The actual clustering is now achieved by drawing the lines of Figs.~\ref{fig:03-1},\ref{fig:03-3} between the particles of the individual clusters such that each cluster
gives a connected diagram.
Two possible ways of building $(1,2,3)(4,5)(6)$ can be seen in Figs.~\ref{fig:03-7}, \ref{fig:03-8}, which both belong to a clustering realisation of
$N=6$ particles, i.e.~all particles are actually represented by dots. If we look back at the diagrams of $\hat{\mathcal{C}}_{\mathrm{tot}}$ in Fig.~\ref{fig:03-5},
we see that all of these can in the same way be understood as belonging to a clustering realisation of some $N$ particles but only the two, three and so on correlated particles 
are actually drawn. All other particles are uncorrelated and must thus be thought of as isolated dots giving factors of unity.
\begin{figure}[htp]
 \centering 
 \includegraphics[scale=0.5]{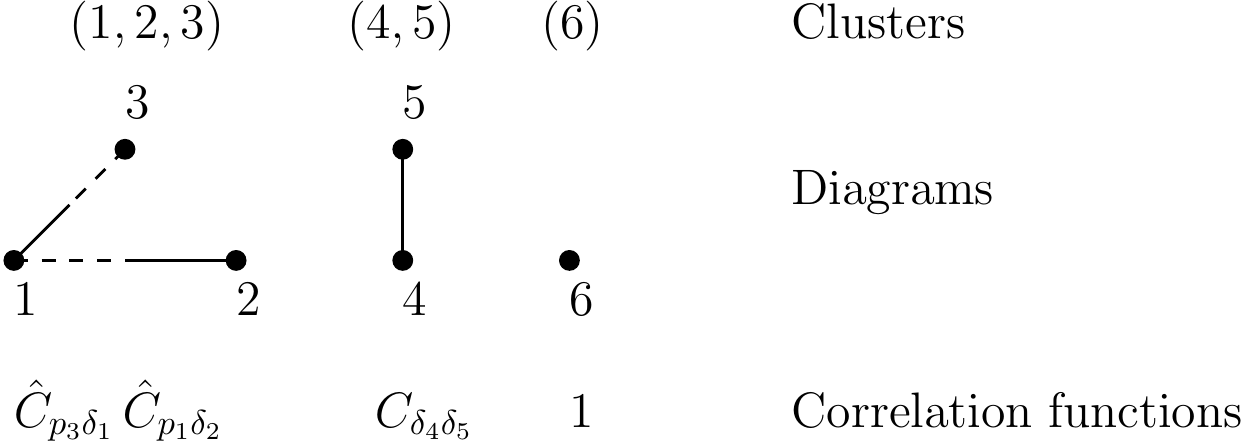}
 \caption{One term in $\mathcal{C}_{\mathrm{tot}}$ contributing to the $(1,2,3)(4,5)(6)$ clustering realisation.}
 \label{fig:03-7}
\end{figure}
\begin{figure}[htp]
 \centering 
 \includegraphics[scale=0.5]{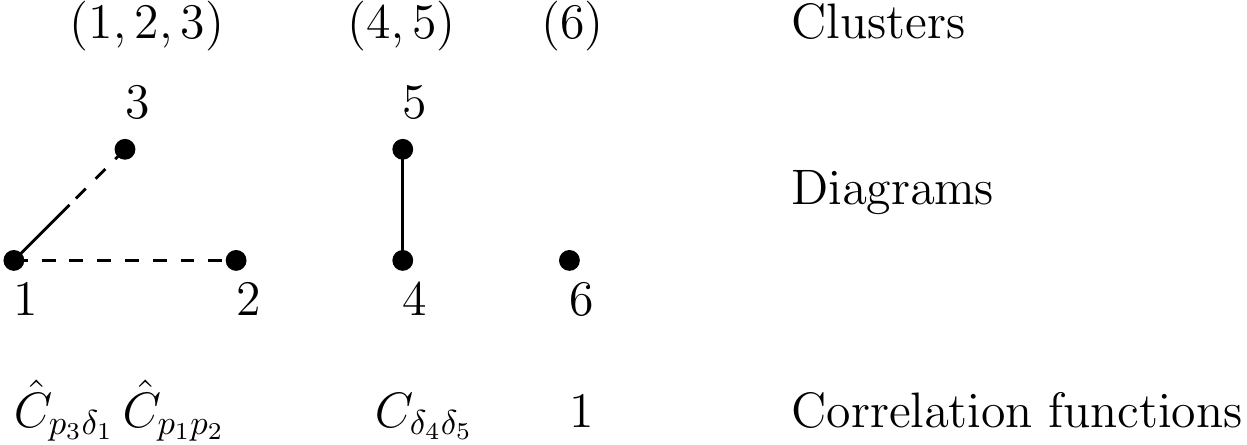}
 \caption{Another term in $\mathcal{C}_{\mathrm{tot}}$ contributing to the $(1,2,3)(4,5)(6)$ clustering realisation.}
 \label{fig:03-8}
\end{figure}
The important aspect is that any diagram in $\hat{\mathcal{C}}_{\mathrm{tot}}$ can be factorized into such clusters even after applying all trace operators.
For the example of Fig.~\ref{fig:03-7} we have
\begin{gather}
 \left(\prod_{j=1}^6 \tr_j\right) \hat{C}_{p_3 \delta_1} \hat{C}_{p_1 \delta_2} C_{\delta_4 \delta_5} \\ 
 = \left(\tr_1 \tr_2 \tr_3 \, \hat{C}_{p_3 \delta_1} \hat{C}_{p_1 \delta_2}\right) \left(\tr_4 \tr_5 \, C_{\delta_4 \delta_5} \right) \left( \tr_6 \right) \;.
\label{eq:03-25}
\end{gather}
For any realisation of a given cluster configuration as a fixed pattern we may thus go through all clusters of the realisation individually while keeping the correlation functions in
the other clusters fixed, sum all connected diagrams compatible with the rules from the previous section for each cluster and then apply the trace operators $\mathrm{Tr}_j$. For the
example of Figs.~\ref{fig:03-7}, \ref{fig:03-8} we would start with the 3-cluster. While holding the $C_{\delta_4 \delta_5}$ in the 2-cluster fixed, we first sum the two
diagrams shown and then continue to sum all remaining connected 3-particle diagrams shown in Figs.~\ref{fig:03-2}, \ref{fig:03-4}. This sum is defined as a single factor
and then held constant while we sum the nine possible 2-particle diagrams of Fig.~\ref{fig:03-5} in the 2-cluster. The 1-particle dot is trivially connected. In general, if we define 
\begin{equation}
 \varSigma_\ell = \frac{1}{\ell!} \, \mathrm{Tr}_{j_1} \, \ldots \, \mathrm{Tr}_{j_\ell} \, \mathcal{C}^{(\ell)}_{\mathrm{con}}\left(\tensgr{\chi}_p(t_\mathrm{i})\right) \;,
\label{eq:03-26}
\end{equation}
where $\mathcal{C}^{(\ell)}_{\mathrm{con}}$ is the sum of all connected $\ell$-point diagrams, then the contribution from all diagrams in $\hat{\mathcal{C}}_{\mathrm{tot}}$
belonging to a specific realisation of an $N$-particle cluster configuration $\{m_\ell\}$ can be factorized as
\begin{equation}
 \prod_{\ell=0}^N \, \left(\ell! \, \varSigma_\ell\right)^{m_\ell} \;.
\label{eq:03-27}
\end{equation}
For example, since one cannot have a disconnected correlation of two particles, $\varSigma_2$ can be seen as the 2-particle sum in Fig.~\ref{fig:03-5}.
The $(1,2,3)(4,5)(6)$ realisation would give the contribution $(1 \cdot 1)^1 \, (2 \cdot \varSigma_2)^1 \, (6 \cdot \varSigma_3)^1$. 
However, we now have to think about about all other possible realisations of a given cluster configuration $\{m_\ell\}$. While they all give the same contribution
\eqref{eq:03-27} due to the trace operators and the summing of all connected diagrams, we still need to know how many of them are actually present in the expansion of
$\mathcal{C}^{(N)}_{\mathrm{tot}}$. In principle there are $N!=6!=720$ such realisations since this is just the question of arranging the particles in a certain order
once we have drawn the clustering pattern and fixed the bijective mapping. However, there are equivalent realisations which lead to overcounting.
We need to take two things into account:

\begin{itemize}
 \item Exchanging the order of particles inside a cluster leaves the sum over all connected diagrams of the cluster invariant. Each particle exchange would thus give a new term
       that has already been accounted for, so we need to divide out all possible orderings within clusters. For a given $N$-particle cluster configuration
       these are $\prod_{\ell=0}^N \left(\ell!\right)^{m_\ell}$.
       
 \item For two clusters of equal size exchanging all particles between them also leaves their respective sums over connected diagrams invariant. For an arbitrary number of equally
       sized clusters this amounts to reordering the sequence of these clusters in the clustering pattern and thus to $m_\ell!$ possibilities.  
\end{itemize}

Accounting for these restrictions, we may now rewrite the traces over $\mathcal{C}^{(N)}_{\mathrm{tot}}$ from \eqref{eq:03-23} as a sum over cluster configurations $\{m_\ell\}$
subject to the constraint \eqref{eq:03-24}. This allows writing the free grand canonical partition functional into
\begin{align}
 Z_{\mathrm{gc},0}[H] &= \sum_{N=0}^{\infty} \, \frac{1}{N!} \, \sum_{\{m_\ell\}'} \, \frac{N!}{\prod_{\ell=0}^N \, (\ell !)^{m_\ell} \, m_\ell!} \, \prod_{\ell=0}^N \, \left(\ell ! \varSigma_\ell\right)^{m_\ell} \notag \\ 
                      &= \sum_{N=0}^{\infty} \, \sum_{\{m_\ell\}'} \, \prod_{\ell=0}^N \, \frac{\left(\varSigma_\ell\right)^{m_\ell}}{m_\ell !} \displaybreak[0] \notag \\
                      &= \sum_{\{m_\ell\}} \, \prod_{\ell'} \, \frac{\left(\varSigma_\ell\right)^{m_\ell}}{m_\ell !} 
                       = \prod_{\ell=0}^{\infty} \, \sum_{m_\ell=0}^{\infty} \, \frac{\left(\varSigma_\ell\right)^{m_\ell}}{m_\ell !} \displaybreak[0] \notag \\
                      &= \prod_{\ell=0}^{\infty} \, \e^{\varSigma_\ell} = \nexp{\sum_{\ell=0}^{\infty} \, \varSigma_\ell} \;.
\label{eq:03-28}
\end{align}
Going from the second to the third line we use that first summing over all cluster configurations with the $N$-particle constraint $\{m_\ell\}'$ and then summing over all
particle numbers is identical to summing over all cluster configurations without this constraint. In the third line $\ell'$ runs over all cluster sizes present in
the preceding cluster configuration. We then reorder the sum not in terms of cluster configurations but in terms of cluster size.

Readers familiar with QFT will recognize the above equation \eqref{eq:03-28} as what is often called \emph{exponentiation of disconnected diagrams}. The reason that we obtain this
familiar structure for the generating functional is that the topological principles behind the Mayer cluster expansion and the Feynman diagrams of QFT are the same. Showing how
these principles can be implemented for the initial correlations between particles sampling a Gaussian random field, rather than for the particle interactions, should be seen as
one of the main achievements of this paper. We also mention that it should in principle be possible to include slightly non-Gaussian fields into this formalism, since higher
order cumulants of the random field would be the equivalent of vertices in our diagrammatic language.

Since it does not matter which actual $\ell$ particles are used in calculating $\varSigma_\ell$ we can always think
of some representative set of $\ell$-particles for which the trace operators in \eqref{eq:03-26} are defined. We include the sources $\tens{J},\tens{K}$
defined for this representative set by modifying the trace operators as
\begin{equation}
 \mathrm{Tr}_j \rightarrow \mathrm{Tr}_j \, \nexp{\ii \, \bsi{t}{t_\mathrm{i}}{t_\mathrm{f}} \left( \vec{J}_j(t) \cdot \vec{x}_j(t) + \vec{K}_j(t) \cdot \vec{\chi}_j(t) \right) } \;.
\label{eq:03-29}
\end{equation}
The complete grand canonical generating functional is now easily obtained by applying the interaction operator
\begin{equation}
 Z_{\mathrm{gc}}[H,\tens{J},\tens{K}] = \e^{\ii \, \hat{\mathrm{S}}_\mathrm{I} } \, \nexp{\sum_{\ell=0}^{\infty} \, \varSigma_\ell[H,\tens{J},\tens{K}]} \;.
\label{eq:03-30}
\end{equation}

\section{Non-interacting cumulants}

\subsection{General form}

Having found the general form of $Z_{\mathrm{gc},0}$, we will now present a scheme how the connected non-interacting correlators of the collective field $\Phi$ can be derived.
For the remainder of this paper we will call these $n$-point cumulants. Their general definition is given by
\begin{align}
 G_{\Phi_{\alpha_1} \ldots \Phi_{\alpha_n}}(1,\ldots,n) &= \fder{}{H_{\alpha_1}(1)} \, \ldots \, \fder{}{H_{\alpha_n}(n)} \notag \\
                                                                     & \quad \left. \ln Z_{\mathrm{gc}}[H,\tens{J},\tens{K}] \, \right\vert_{H=\tens{J}=\tens{K}=0} \;,
\label{eq:04-1}
\end{align}
where we use the notation $1 = (t_1,\vec{k}_1)$ to bundle time and Fourier space coordinates into an `external label'.
We further shorten the notation by understanding that $\Phi_{\alpha_1}$ means the collective field of type $\alpha$ evaluated at the external label $1$.
In the non-interacting case of $\hat{S}_{\mathrm{I}} = 0$ the exponential form of $Z_{\mathrm{gc},0}$ cancels against the logarithm. Due to the ordering by the
number $\ell$ of representative particles it will be advantageous to also separate $G^{(0)}_{\Phi_{\alpha_1} \, \ldots \, \Phi_{\alpha_n}}$ into its contributions
from $\ell$ particles as
\begin{align}
 G^{(0)}_{\Phi_{\alpha_1} \ldots \Phi_{\alpha_n}} &= \left. \fder{}{H_{\alpha_1}} \ldots \fder{}{H_{\alpha_n}} \left( \sum_{\ell=0}^\infty \, \varSigma_\ell[H,\tens{J},\tens{K}] \right) \, \right\vert_{H=\tens{J}=\tens{K}=0} \notag \\
                                                &= \left. \sum_{\ell=0}^\infty \, \usi{\Gamma_{\mathrm{i}}} \hat{\mathcal{C}}^{(\ell)}_{\mathrm{con}} \, \hat{\Phi}^{(\ell)}_{\alpha_1}  \ldots \hat{\Phi}^{(\ell)}_{\alpha_n} \, Z^{(\ell)}_0[\tens{J},\tens{K}] \, \right\vert_{\tens{J}=\tens{K}=0} \notag \\ 
                                                &= \sum_{\ell=0}^\infty G_{\Phi_{\alpha_1} \, \ldots \, \Phi_{\alpha_n}}^{(0,\ell)}
\label{eq:04-2}
\end{align}
Observe that all collective quantities are now intrinsically defined for the $\ell$ representative particles that appear in
$\varSigma_\ell$. In the second line we returned to expressing the collective fields as operators whose single-particle contributions in Fourier space read (cf. Eqs.~(54),(58) in
\citep{Bartelmann2014})
\begin{align}
 \hat{\phi}_j(\vec{k}_1,t_1) &= \cvector{ \hat{\phi}_{\rho_j}(1) \\ \hat{\phi}_{B_j}(1) } = \cvector{ \hat{\phi}_{\rho_j}(1) \\ \hat{b}_j(1) \, \hat{\phi}_{\rho_j}(1) }\notag \\
                             &= \cvector{ \nexp{-\ii \, \vec{k}_1 \cdot \fder{}{\vec{J}_{q_j}(t_1)} } \\ \left(\ii \, \vec{k}_1 \cdot \fder{}{\vec{K}_{p_j}(t_1)} \right) \nexp{-\ii \, \vec{k}_1 \cdot \fder{}{\vec{J}_{q_j}(t_1)} } }
\label{eq:04-3}
\end{align}
acting on $Z^{(\ell)}_0[\tens{J},\tens{K}]$ which is the deterministic canonical generating functional \eqref{eq:03-18}, only defined for the $\ell$ representative particles.
Furthermore, we again used (\ref{eq:03-19}) to express the factors of $\tensgr{\chi}_p$ in the initial correlations as functional derivatives leading to the operator
$\hat{\mathcal{C}}^{(\ell)}_{\mathrm{con}}$ which absorbed the factor $\frac{1}{\ell!}$ in (\ref{eq:03-26}). The integration over initial conditions is only over the ideal gas part
\begin{equation}
 \usi{\Gamma_{\mathrm{i}}^{(\ell)}} = \mpd^\ell \usi{\tens{q}^{(\ii)}} \usi{\tens{p}^{(\ii)}} P_{\sigma_p}^{\mathrm{MB}}(\tens{p}^{(\ii)}) \;,
\label{eq:04-4}
\end{equation}
where the quantities $\tens{q}^{(\ii)}$ and $\tens{p}^{(\ii)}$ hold the initial positions and momenta of the $\ell$ representative particles. It is also important to
notice that by specifying a maximum order of correlations to be taken into account, one can effectively truncate the series in \eqref{eq:04-2}. This is due to the fact that
a maximum order of correlations translates into a maximum number of lines in the diagrams making up $\hat{\mathcal{C}}^{(\ell)}_{\mathrm{con}}$. Since these diagrams are connected
one needs at least $\ell-1$ correlation lines to connect $\ell$ particles. A maximum number $n$ of correlation lines thus means a truncation at $\ell = n + 1$ particles.

\subsection{Effects of collective field and initial correlation operators}\label{ef_op}

The effects of collective field operators were derived in \citet{Bartelmann2014}. One particular advantage of the operator approach is that one may freely choose the order in which
the various functional derivatives are to be applied. Since the $B$ field operator factorizes as $\hat{b}_j(1) \, \hat{\phi}_{\rho_j}(1)$ it is advantageous
to first calculate density-only cumulants since mixed cumulants between density and response fields can be obtained from them by multiplying with appropriate $b(1)$ prefactors.

We adopt the notion of a particle $j$ `carrying' an external label $1$ if a single particle operator $\hat{\phi}_j$ is applied at label 1. Physically this
represents the contribution of the particle $j$ to the Fourier mode $\vec{k}_1$ of the collective field vector $\Phi$ at the time $t_1$. The effect of applying multiple
operators $\hat{\phi}_{\rho_j}(1)$ all belonging to the same particle $j$ but with different external labels will result in a shift of the $\tens{J}$ source,

\begin{equation}
 \hat{\phi}_{\rho_j}(1) \, \ldots \, \hat{\phi}_{\rho_j}(s_j) \, Z^{(\ell)}_0[\tens{J},\tens{K}] = Z_0^{(\ell)}[\tens{J}+\tens{L}_j,\tens{K}] \;,
\label{eq:04-5}
\end{equation}
with the shift tensor defined as
\begin{equation}
 \tens{L}_j(t) = - \sum_{\{s_j\}} \, \dirac\left(t-t_{s_j}\right) \, \cvector{ \vec{k}_{s_j} \\ 0 } \otimes \vec{e}_j 
\label{eq:04-6}
\end{equation}
and $\{s_j\}$ is the set of external labels carried by the particle $j$. For more than one particle, the shift tensors are added,
\begin{equation}
 \tens{L}(t) := \sum_{j=1}^\ell \, \tens{L}_j(t) \;.
\label{eq:04-7}
\end{equation}
Having applied all $\hat{\phi}_{\rho_j}$ for all particles present, we need to account for the effects of applying $\hat{b}_j(1)$ and the functional derivatives in
$\mathcal{C}_{\mathrm{con}}^{(\ell)}$ to $Z_0^{(\ell)}$. Application of $\hat{b}_j(1)$ will result in a $b_j(1)$ factor:
\begin{gather}
 \hat{b}_j(1) \, Z_0^{(\ell)}[\tens{J} + \tens{L},\tens{K}] = b_j(1) \, Z_0^{(\ell)}[\tens{J} + \tens{L},\tens{K}] \;, \notag \\
 b_j(1) = \left(-\ii \, \bsi{t}{t_\mathrm{i}}{t_\mathrm{f}} \scalar{\tens{J}(t) + \sum_{a=1}^\ell \, \tens{L}_a(t) }{G(t,t_1) \, \cvector{ 0_3 \\ \vec{k}_1} \otimes \vec{e}_j } \right) \;.
\label{eq:04-8}
\end{gather}
Once all derivatives have been applied we turn off the sources $\tens{J},\tens{K}$ and the factor reduces substantially to
\begin{equation}
 b_j(1) = \ii \, \vec{k}_1 \cdot \sum_{\{s_j\}} \, \vec{k}_{s_j} \, \g_{qp}(t_{s_j},t_1) \;.
\label{eq:04-9}
\end{equation}
In complete analogy we can define factors for the result of the application of the initial correlation operators.
\begin{equation}
 c_{\delta_i p_j} = -\ii \, \vec{C}_{\delta_i p_j} \cdot \sum_{\{s_j\}} \, \vec{k}_{s_j} \, \g_{qp}(t_{s_j},t_\mathrm{i}) \;,
\label{eq:04-10}
\end{equation}
\begin{align}
 c_{p_i p_j} = \exp\left\{- \left( \sum_{\{s_i\}} \vec{k}_{s_i} \, \g_{qp}(t_{s_i},t_\mathrm{i}) \right)^\top C_{p_i p_j} \right. \notag \\
               \left. \, \left( \sum_{\{s_j\}} \vec{k}_{s_j} \, \g_{qp}(t_{s_j},t_\mathrm{i}) \right) \right\} - 1 \;.
\label{eq:04-11}
\end{align}
Finally, the free non-averaged generating functional reduces to
\begin{align}
 \left. Z_0^{(\ell)}[\tens{J},\tens{K}] \right\vert_{\tens{J}=\tens{K}=0} =& \, \nexp{\ii \, \bsi{t}{t_\mathrm{i}}{t_\mathrm{f}} \scalar{\tens{L}(t)}{\mathcal{G}(t,t_\mathrm{i}) \, \tens{x}^{(\ii)}}} \notag \\
 =& \, \prod_{j=1}^\ell \, \nexp{-\ii \, \left(\sum_{\{s_j\}} \, \g_{qq}(t_{s_j},t_\mathrm{i}) \, \vec{k}_{s_j} \right) \cdot \vec{q}^{\,(\ii)}_j } \notag \\
  & \,\nexp{-\ii \, \left( \sum_{\{s_j\}} \, \g_{qp}(t_{s_j},t_\mathrm{i}) \, \vec{k}_{s_j} \right) \cdot \vec{p}^{\,(\ii)}_j } \;.
\label{eq:04-12}
\end{align}
Only two components of the free single particle propagator $G(t,t')$ are needed namely $\g_{qq}$ and $\g_{qp}$. In most systems we
have $\g_{qq}(t,t_\mathrm{i}) = \Theta(t-t_\mathrm{i}) = 1$ since we are only interested in times $t \geq t_\mathrm{i}$. Thus only $\g_{qp}$ remains. For convenience of notation
we now define
\begin{equation}
 \g_{12} \coloneqq \g_{qp}(t_1, t_2) \quad \text{and} \quad \g_1 \coloneqq \g_{qp}(t_1,t_\mathrm{i}) \;.
\label{eq:04-13}
\end{equation}
Concerning the $b_j$ factors we can now derive a straightforward theorem that will help us reduce the number of terms that we need to calculate later. 

\begin{theorem}
 If in a term contributing to some $G_{\Phi_{\alpha_1} \ldots \Phi_{\alpha_n}}^{(0,\ell)}$ a particle `$j$' carries only external labels belonging to
 $B$-fields, then the term vanishes.
\label{theo:04-1}
\end{theorem}

\begin{proof}
This is best done iteratively, beginning with a particle carrying only a single $B$-field label $1$.
Then immediately $b_j(1) = \ii \, k_1^2 \, \g_{11} = 0$ due to $\g_{11} = \g_{qp}(t_1,t_1) = 0$. For two $B$-fields $B(1), B(2)$ we have
\begin{align*}
  b_j(1) \, b_j(2) &= \left(\ii \, \vec{k}_1 \cdot \vec{k}_2 \, g_{21} \right) \, \left( \vec{k}_2 \cdot \vec{k}_1 \, g_{12} \right) \notag \\
                   &\propto \Theta(t_2-t_1) \, \Theta(t_1-t_2) \;.
\end{align*}
The only possibility for both Heaviside functions not to vanish would be $t_1=t_2$. But this leads to factors $g_{qp}(t_1,t_1) = g_{qp}(t_2,t_2) = 0$. 
   
In the general case of $n$ $B$-fields we use a diagrammatic argument. Picture every time coordinate included in the external labels as a point. Due to \eqref{eq:04-9} 
a factor $b_j(1)$ can be seen as a sum of lines between $t_1$ and all other $n-1$ time coordinates $t_i$ representing the $g_{i 1}$. The product $b_j(1) \, b_j(2) \ldots b_j(n)$
then results in a sum of $(n-1)^n$ terms each containing $n$ lines connecting all $n$ instances in time. As in the case $n=2$, any pair of points connected by two lines,
i.e.~a closed loop contributes zero. The argument for a 2-point loop is easily extended to a general $m$-point loop as
\begin{align*}
 \Theta(t_{i_1} - t_{i_2}) \, \Theta(t_{i_2}-t_{i_3}) &\ldots \Theta(t_{i_{m-1}}-t_{i_m}) \, \Theta(t_{i_m}-t_{i_1}) \notag \\ \notag \\ \Longrightarrow t_{i_m} \geq & \; t_{i_1} \geq t_{i_2} \geq \ldots \geq t_{i_m}
\end{align*}
Again this can only be satisfied if all time coordinates are identical which leads to vanishing propagators. But as we need to connect $n$ instances in time with $n$ lines
in each term of the product $b_j(1) \, b_j(2) \ldots b_j(n)$ there is no possibility of doing so without creating a closed loop and thus all terms vanish. 
\end{proof}
A straightforward corollary that we can immediately derive from this theorem is the following:
\begin{corollary}
Any non-interacting $B$-field-only cumulant vanishes: $G_{B_1 \, \ldots \, B_n}^{(0,\ell)} = 0$.
\label{theo:04-2}
\end{corollary}
The above Theorem \ref{theo:04-1} is merely a consequence of the theory respecting the causality of interactions. We can derive yet another theorem.
\begin{theorem}
 If in a term contributing to some $G_{\Phi_{\alpha_1} \, \ldots \, \Phi_{\alpha_n}}^{(0,\ell)}$ one of the particles involved, say particle `$a$', carries no external labels,
 i.e.~no $\hat{\phi}_{\rho_a}$ has been applied, then the term vanishes.
\label{theo:04-3}
\end{theorem}
\begin{proof}
 Since the particle $a$ is involved in the cumulant it must be connected in a diagrammatic sense with one of the three line types of our diagrammatic representation.
 
 If the particle is `inside the diagram', i.e.~connected to more than one line, then the diagram rules demand that at least one of the lines connecting to $a$ is
 of the dashed $p$-type. Thus either a factor of $c_{\delta_i p_a}$ or $c_{p_i p_a}$ is present. But since $a$ carries no external labels the set $\{s_a\}$ in \eqref{eq:04-10},
 \eqref{eq:04-11} is empty and thus $c_{\delta_i p_a} = 0 = c_{p_i p_a}$.
 
 If the particle is `at the boundary of the diagram', i.e.~connected to only one line, then we need to distinguish three different cases.
 
 \begin{itemize}
  \item The particle is connected to a dashed $p$-type line. The same argument as above applies.  
  \item It is connected by a $C_{\delta_a \delta_j}$ line. Since $\{s_a\}$ is empty we have for the factor from the free generating functional
        \begin{equation*}
         \nexp{-\ii \, \left(\sum_{\{s_a\}} \,\vec{k}_{s_a} \right) \cdot \vec{q}^{\,(\ii)}_a } = \exp(0) = 1
        \end{equation*}
        and thus the only quantity left that depends on $\vec{q}^{\,(\ii)}_a$ is $C_{\delta_a \delta_j}$. This leaves us with
        \begin{equation*}
         \umi{3}{q^{(\ii)}_a} C_{\delta_a \delta_j} = \umi{3}{q^{(\ii)}_a} \, \xi\left(\vec{q}^{\,(\ii)}_a,\vec{q}^{\,(\ii)}_j \right) = 0
        \end{equation*}
        since we define our 2-point correlation function $\xi$ as the Fourier transform of a power spectrum $P_0(k)$ which vanishes at $k=0$.        
  \item It is connected to the solid $\delta$-side of a $\vec{C}_{\delta_a p_i}$. With the arguments from the previous case we have
        \begin{equation*}
         \umi{3}{q^{(\ii)}_a} \vec{C}_{\delta_a p_i} = \left\langle \left( \umi{3}{q^{(\ii)}_a} \delta\left(\vec{q}^{\,(\ii)}_a\right) \right) \vec{p}\left(\vec{q}^{\,(\ii)}_i\right) \right\rangle = 0
        \end{equation*}
        since the density contrast must obey particle conservation. Notice that $\vec{p}\left(\vec{q}^{\,(\ii)}_i\right)$ is the initial momentum field evaluated at the
        initial position of the particle $i$.
 \end{itemize}
\end{proof}
We can now combine Theorems \ref{theo:04-1} and \ref{theo:04-3} to derive another corollary.
\begin{corollary}
  For all $\ell > n$, the cumulant $G^{(0,\ell)}_{\rho_1 \, \ldots \, \rho_n B_{n+1} \, \ldots \, B_{n+m}} = 0$ .
\label{theo:04-4}
\end{corollary}
\begin{proof}
 We begin by considering density-only cumulants first. These are given by
 \begin{align}
  G^{(0,\ell)}_{\rho_1  \ldots \rho_n} = \usi{\Gamma_\mathrm{i}} \mathcal{C}^{(\ell)}_{\mathrm{con}} \,  &\left( \sum_{j=1}^\ell \, \hat{\phi}_{\rho_j}(1) \right) \ldots \notag \\
  & \ldots \left( \sum_{j=1}^\ell \, \hat{\phi}_{\rho_j}(\mathrm{n}) \right) \, \left. Z^{(\ell)}_0[\tens{J},\tens{K}] \, \right\vert_{\tens{J}=\tens{K}=0} \;.
  \label{eq:04-14}
 \end{align}
 Expanding the product of operators we get a sum where in each term every external label appears exactly once, i.e.~each term consists of $n$ factors.
 For $\ell > n$, there will be at least one particle in each term that will not carry an external label and thus all terms vanish according to Theorem \ref{theo:04-3}. Next we can add any number $m$ of $B$-fields to the
 correlator. This leads to a sum of terms with $n+m$ factors. But as $\ell > n$ for every term we still have at least one particle which either carries no external label
 which causes the term to vanish or it only carries $B$-field indices which also makes the term vanish according to Theorem \ref{theo:04-1}. 
\end{proof} 
It is also interesting to note that the scaling of any cumulant $G^{(0,\ell)}_{\Phi_{\alpha_1} \, \ldots \, \Phi_{\alpha_n}}$ with the mean particle density $\mpd$
is only controlled by the number of representative particles due to $\varSigma_\ell \propto \mpd^\ell$. This relates to the fact that in a continuous fluid picture
of the density field any density-only $n$-point cumulant $G^{(0)}_{\rho_1 \, \ldots \, \rho_n}$ has only terms scaling with $\mpd^n$. In our particle picture
we have terms of all possible scalings $\mpd^\ell$ for $1 \leq \ell \leq n$ due to \eqref{eq:04-2}. All terms with $\ell < n$ may consequently be interpreted as shot noise terms where
some of the $n$ particles have been identified with each other, thus reducing the amount of possible initial correlation.

We can also consider an $n$-point density-only cumulant with $\ell=n$ and then change any number of density fields into response fields. According to Corollary \ref{theo:04-4}
the resultant cumulant will vanish unless we also reduce the number of particles $\ell$ by one for each $B$-field. Just as in the canonical ensemble we thus see that for fixed
$n$-point order the response field leads to an identification of particles since it encodes how the effects of interactions are transported forward through time by single particles
which then may contribute to other collective fields at a later time.

\subsection{Hierarchy of label combinatorics}\label{org_comb}

A general density-only $\ell$-particle $n$-point cumulant has the form given in \eqref{eq:04-14}. We first multiply out all the sums into individual terms.
Next we organise them in a hierarchical fashion.

The first level of the hierarchy should tell us how many of the $n$ external labels are carried by each of the $\ell$ particles. We will call this category a `label distribution' 
$\#_1|\ldots|\#_\ell$. An example for $\ell=2$ would be $1|2$ where one particle carries one label and the other two labels. Note that we do not ask which particle carries the single
label, both possibilities are in the same distribution and the same holds for higher $\ell$.

The next lower level category is called a `label grouping' $(\{s_1\}; \ldots ; \{s_\ell\})$ specifying which specific labels are grouped due to being carried by
the same particle. Considering an $\ell=3$ particle $n=4$-point cumulant a possible label grouping would be $(1;2;3,4)$ as well as $(1;3;2,4)$ and both belong to the label
distribution $1|1|2$.

The lowest category is a `labeling' representing a single term in \eqref{eq:04-14}. For the $\ell=3$ particle $n=4$-point cumulant with particles
$a,b,c$ one could realise $(1;2;3,4)$ as e.g.~$a_1 b_2 c_3 c_4$, $b_1 c_2 a_3 a_4$, $a_1 c_2 b_3 b_4$ and so on. The hierarchy for this example can be seen in Fig.~\ref{fig:04-1}.
\begin{figure}[htp]
 \centering
 \includegraphics[scale=0.5]{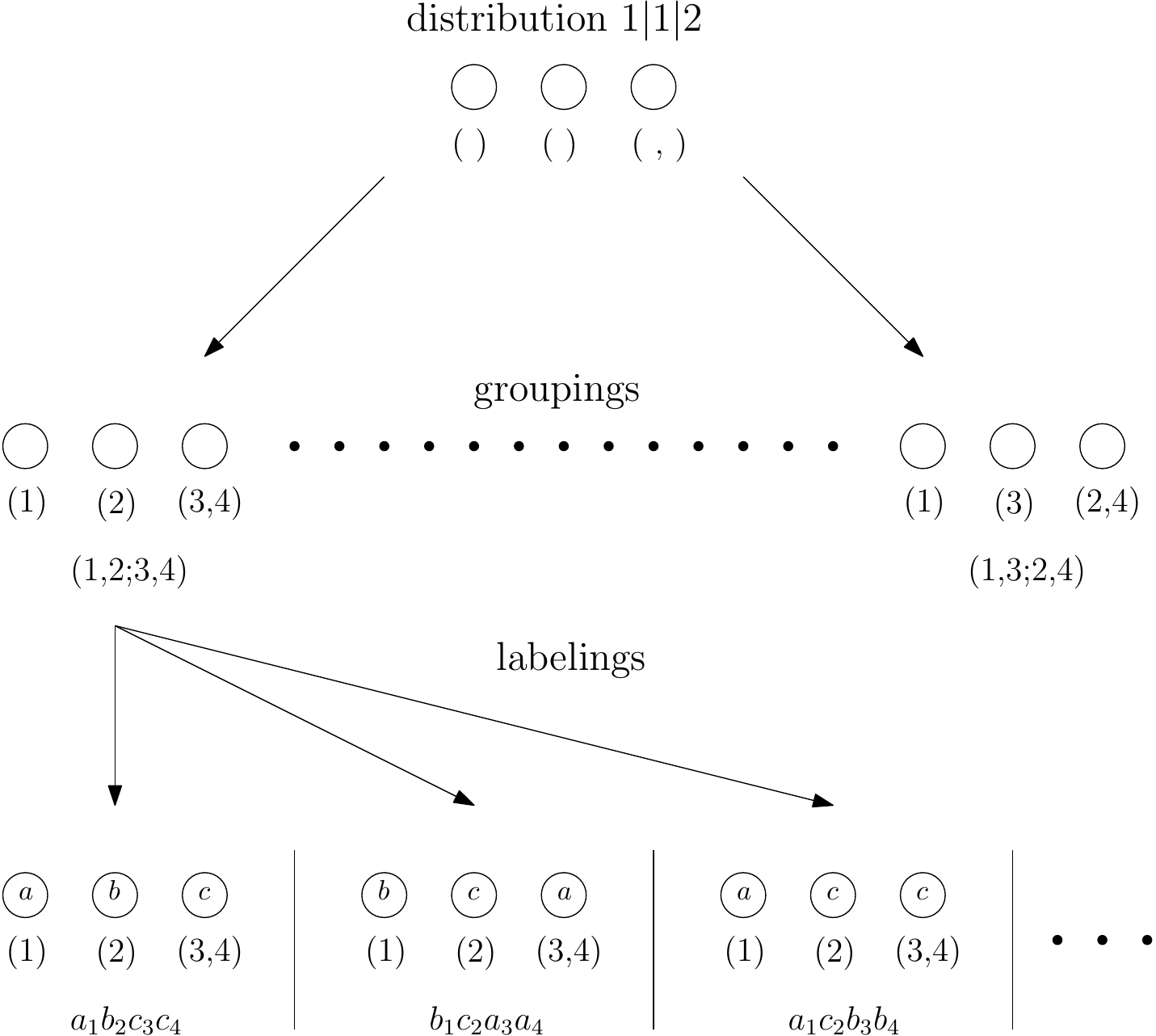}
 \caption{The hierarchy of external labels.}
\label{fig:04-1}
\end{figure}
The important point to realise is that the integral over initial momenta will result in a factor which only depends on the label grouping and not on the individual labeling.
This factor turns out to be a damping factor of the form
\begin{gather}
 \mathfrak{D}_{(\{s_1\}; \ldots ; \{s_\ell\})} := \prod_{j=1}^\ell \, \nexp{-\frac{\sigma_p}{2} \left( \sum_{\{s_j\}} \g_{s_j} \vec{k}_{s_j} \right)^2 } \notag \\
                                                = \usi{\tens{p}^{(\ii)}} \prod_{j=1}^\ell \, P^{\mathrm{MB}}_{\sigma_p}(\vec{p}^{\,(\ii)}_j) \, \nexp{-\ii \, \left( \sum_{\{s_j\}} \g_{s_j} \vec{k}_{s_j} \right) \cdot \vec{p}^{\,(\ii)}_j } \;.
\label{eq:04-15}
\end{gather}
It is thus a good idea to organize cumulants in terms of label groupings. The general strategy for calculating any density-only cumulant can thus be summarized as:
\begin{itemize}
 \item Find all possible label distributions $\#_1 | \ldots | \#_\ell$ where all particles carry at least one label (distributions where one particle does not carry a label vanish due
       to Theorem \ref{theo:04-3}).
 
 \item For each of these identify all possible label groupings and pick an arbitrary representative.
       
 \item For this representative label grouping go through all possible labelings, execute the integral over the initial particle positions taking into account diagrams
       of initial correlations up to a desired order. Gather the results into one function $\mathcal{T}^{(\ell)}_{\#_1 | \ldots | \#_\ell}$. The complete contribution from a
       label distribution $\#_1 | \ldots | \#_\ell$ is obtained by summing $\mathcal{T}^{(\ell)}_{\#_1 | \ldots | \#_\ell}$ evaluated with all groupings of external labels belonging
       to the distribution.
\end{itemize}
This has the advantage that for any distribution one only needs to calculate a single function. Another benefit of this approach is 
that we directly see how labels of the collective fields are grouped onto particles, allowing us to drop certain terms according to Theorem \ref{theo:04-1}. We will see how
this works in detail when we compute the general $\ell=2$ particle cumulants.

\subsection{Explicit form of initial correlations}

Since the primary motivation for this work is cosmological structure formation and because we want to be able to give explicit expressions for the cumulants we now
adopt the specific form of the 2-particle correlations given in \citep{Bartelmann2014}. The initial density field is characterized by the power spectrum of its contrast in the
standard way as
\begin{equation}
 \ave{\delta\left(\vec{k}_1\right) \, \delta\left(\vec{k}_2\right)} = (2\pi)^3  \dirac\left(\vec{k}_1 + \vec{k}_2\right) \, P_0(k_1) \;.
\label{eq:04-16}
\end{equation}
The initial momentum field is assumed to be irrotational so there exists a momentum potential $\psi$ with
\begin{equation}
 \vec{p}^{\,(\ii)}_j = \alpha \left. \nabla \psi\left( \vec{r} \right) \right\vert_{\vec{r} = \vec{q}^{\,(\ii)}_j} \;.
\label{eq:04-17}
\end{equation}
The continuity equation demands that this potential is linked to the initial density contrast as
\begin{equation}
 \delta_j = \delta(\vec{q}_{j,0}) = - \left. \left( \nabla^2 \, \psi(\vec{r}) \right) \right\vert_{\vec{r}=\vec{q}^{\,(\ii)}_j} \;.
\label{eq:04-18}
\end{equation}
This allows us to express all three types of correlations as Fourier transforms of the density power spectrum. Due to statistical homogeneity
and isotropy they will only depend on the modulus of the separation vector of the two particles involved.
\begin{equation}
 C_{\delta_i \delta_j} = \fmi{3}{h} \e^{\ii \, \vec{h} \cdot \left(\vec{q}^{\,(\ii)}_i - \vec{q}^{\,(\ii)}_j \right)} \, P_0(h)
\label{eq:04-19}
\end{equation}
\begin{equation}
 \vec{C}_{\delta_i p_j} = -\ii \, \alpha \, \fmi{3}{h} \e^{\ii \, \vec{h} \cdot \left(\vec{q}^{\,(\ii)}_i - \vec{q}^{\,(\ii)}_j \right)} \, \frac{\vec{h}}{h^2} \, P_0(h)
\label{eq:04-20}
\end{equation}
\begin{equation}
 C_{p_i p_j} = \alpha^2 \fmi{3}{h} \e^{\ii \, \vec{h} \cdot \left(\vec{q}^{\,(\ii)}_i - \vec{q}^{\,(\ii)}_j \right)} \, \frac{\vec{h} \otimes \vec{h}}{h^4} \, P_0(h)
\label{eq:04-21}
\end{equation}

\subsection{1-particle cumulants}

The 1-particle cumulants $G_{\Phi_{\alpha_1} \, \ldots \, \Phi_{\alpha_n}}^{(0,1)}$ can be written down directly to any desired order in $n$ as was already shown in
\citep{Das2012,Mazenko2010}. Initial correlations have no effect as $\mathcal{C}^{(1)}_{\mathrm{con}} = 1$ in \eqref{eq:04-14}. Since there is only one particle,
the only possible grouping of external labels is their entirety, i.e.~$\{s_a\}=\{1,\ldots,n\}$. This means that integrating the spatial part of \eqref{eq:04-12} over
the initial position of the single particle $a$ gives
\begin{equation}
\umi{3}{q^{(\ii)}_a} \nexp{-\vec{q}^{\,(\ii)}_a \cdot \sum_{\{s_a\}} \vec{k}_{s_a}} = (2\pi)^3 \dirac\left( \sum_{s=1}^n \, \vec{k}_s \right) \;.
\label{eq:04-22}
\end{equation}
Combining this with the Gaussian cutoff, the 1-particle contribution to the $n$-point density-only cumulant is given by
\begin{equation}
 G_{\rho_1  \ldots \rho_n}^{(0,1)} = \mpd \, (2\pi)^3 \dirac\left( \sum_{s=1}^n \, \vec{k}_s \right) \nexp{-\frac{\sigma_p^2}{2} \sum_{s=1}^n \, \g_{s} \, \vec{k}_s} \;.
\label{eq:04-23}
\end{equation}
Any cross correlator between $\rho$ and $B$ is obtained by applying the appropriate $b_a(s) = b(s)$ factor 
\begin{equation}
 G_{\rho_1 \ldots \rho_n B_{n+1} \ldots B_{n+m}}^{(0,1)} = b(n+1) \, \ldots \, b(n+m) \, G_{\rho_1 \, \ldots \, \rho_n}^{(0,1)} \;.
\label{eq:04-24}
\end{equation}
The most interesting cases are the one and two point cumulants.
\begin{equation}
 G^{(0,1)}_{\rho_1} = \mpd \, (2\pi)^3 \dirac\left(\vec{k}_1\right) \qquad G^{(0,1)}_{B_1} = 0 
\label{eq:04-25}
\end{equation}
\begin{gather}
 G^{(0,1)}_{\rho_1 \rho_2} = \mpd \, (2\pi)^3 \dirac\left(\vec{k}_1 + \vec{k}_2\right) \, \nexp{-\frac{\sigma_p^2}{2} \, k_1^2 (\g_1-\g_2)^2} \notag \\
 G^{(0,1)}_{\rho_1 B_2}    = -\ii \, k_1^2 \, \g_{12} G^{(0,1)}_{\rho_1 \rho_2} \notag \\
 G^{(0,1)}_{B_1 \rho_2}    = -\ii \, k_1^2 \, \g_{21} G^{(0,1)}_{\rho_1 \rho_2} \notag \\
 G^{(0,1)}_{B_1 B_2}       = 0
\label{eq:04-26}
\end{gather}
As expected the 1-point correlator of the density field just gives the mean particle density. The interpretation of the 2-point density cumulant as an exponentially damped shot-noise
contribution was motivated at the end of section \ref{ef_op}. In our next work \citep{Fabis2015a}, we will show that $G^{(0,1)}_{\rho B}$ and
$G^{(0,1)}_{B \rho}$ can be understood as propagators for the density field $\rho$ in a statistical sense.

\subsection{2-particle cumulants}

The two-particle correlation operator $\mathcal{C}^{(2)}_{\mathrm{con}}$ is the argument of the two-particle sum seen in Fig.~\ref{fig:03-5}
missing only a prefactor of $2!^{-1}$. It is advantageous to note that the general expression \eqref{eq:04-14}
for a density-only cumulant is by construction invariant under particle exchange or renumbering. We use this to combine diagrams that transform into one another under such
renumbering, i.e.~those diagrams that have the same non-invariant topology. For the present $\ell=2$ case this is shown in Fig.~\ref{fig:04-1}.
Next we have to think about which of the diagrams in Fig.~\ref{fig:03-5} we want to include.
We aim to take into account all terms up to second order in the initial power spectrum $P_0$. This means that all but the last diagram must be considered.
The $\hat{C}_{p_a p_b}$-line represents an exponential function of $C_{p_a p_b} \propto P_0$, while the other two line types are linear in $P_0$.
\begin{figure}[htp]
 \centering
 \includegraphics[scale=0.6]{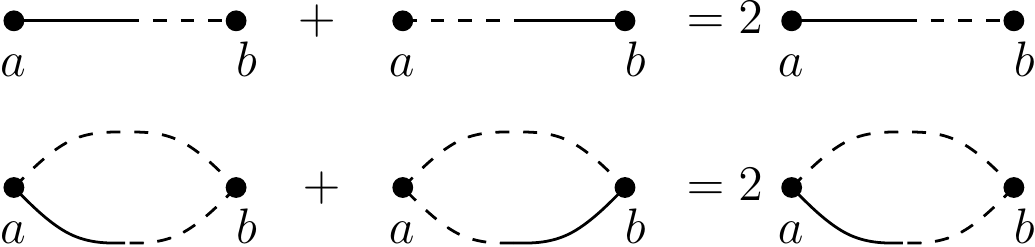}
 \caption{Combining diagrams by particle exchange/relabeling.}
\label{fig:04-2}
\end{figure}
We thus expand the isolated $\hat{C}_{p_a p_b}$-line up to second order which leads to an expansion of the $c_{p_a p_b}$ factor of \eqref{eq:04-11} as
\begin{align}
 c_{p_a p_b} &= -\left( \sum_{\{s_a\}} \vec{k}_{s_a} \, \g_{s_a} \right)^\top C_{p_a p_b} \left( \sum_{\{s_b\}} \vec{k}_{s_b} \, \g_{s_b} \right) \notag \\
             & \; +\frac{1}{2} \left[ \left( \sum_{\{s_a\}} \vec{k}_{s_a} \, \g_{s_a} \right)^\top C_{p_a p_b} \left( \sum_{\{s_b\}} \vec{k}_{s_b} \, \g_{s_b} \right) \right]^2 + \ldots
\label{eq:04-27}
\end{align}
All other dashed lines inside a larger diagram only represent the linear term of \eqref{eq:04-27}. Overall, this leads to the contributions shown in Figs.~\ref{fig:04-3} and
\ref{fig:04-4}. The simple dashed $\hat{C}_{p_a p_b}$-line in both orders only represents the respective first and second order term of its expansion.
\begin{figure}[htp]
 \centering
 \includegraphics[scale=0.6]{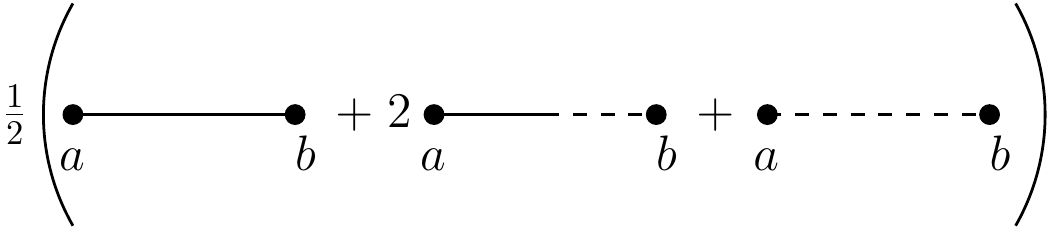}
 \caption{First order in $P_0$ contribution to $\mathcal{C}^{(2)}_{\mathrm{con}}$.}
\label{fig:04-3}
\end{figure}
\begin{figure}[htp]
 \centering
 \includegraphics[scale=0.5]{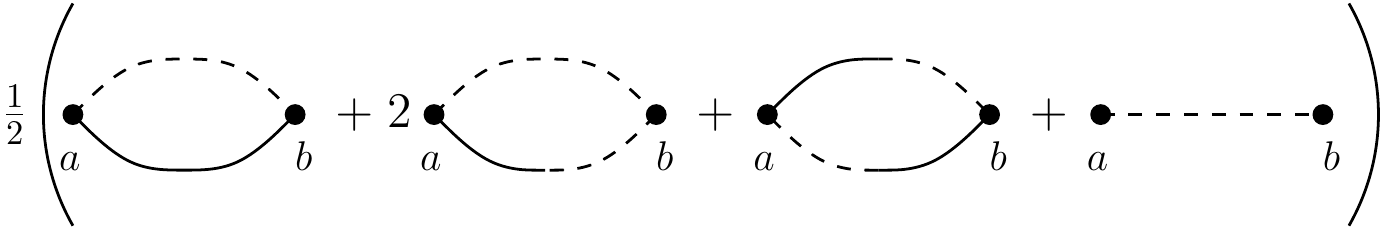}
 \caption{Second order in $P_0$ contribution to $\mathcal{C}^{(2)}_{\mathrm{con}}$.}
\label{fig:04-4}
\end{figure}

In order to write the extensive expressions for the contribution of the individual diagrams to the density-only cumulant in a more compact form we introduce the following shorthand
notation for sums of wavevectors over label sets
\begin{equation}
\vec{k}_{\{s\}} \coloneqq \sum_{\{s\}} \vec{k}_s \;.
\label{eq:04-28}
\end{equation}
Since they will show up in the factors coming from initial correlations, e.g.~\eqref{eq:04-27}, we also define sums of the combination of a free propagator and a wavevector in
accordance with \citep{Bartelmann2014b} as
\begin{equation}
 \vec{K}_{\{s\}} \coloneqq \sum_{\{s\}} \vec{K}_s = \sum_{\{s\}} g_s \vec{k}_s \;.
\label{eq:04-29}
\end{equation}
These quantities are not to be confused with the source terms $\vec{K}_j$ for individual particles which have been turned off by now.

We are now ready to implement the strategy for organising the different terms contributing to an $n$-point cumulant as discussed in section \ref{org_comb}.
For $\ell=2$ particles with indices $a$ and $b$, all non-vanishing labeling distributions are given by $m|(n-m)$ with
$1 < m \leq \lfloor n/2 \rfloor$ and all label groupings $(\{s_1\};\{s_2\}) = \{i_1,\ldots,i_m ; j_1,\ldots,j_{n-m}\}$ contain only two labelings
$a_{i_1} \ldots a_{i_m} b_{j_1} \ldots b_{j_{n-m}}$ and $a_{j_1} \ldots a_{j_{n-m}} b_{i_1} \ldots b_{i_m}$.
Let us pick such a general label grouping $(\{s_1\};\{s_2\})$ and start with the simplest diagram, the $C_{\delta_a \delta_b}$-line in Fig.~\ref{fig:04-2}.
The damping function $\mathfrak{D}_{(\{s_1\};\{s_2\})}$ will be omitted for now. As the diagram is invariant
under particle exchange both labelings contained in any label grouping give the same result and thus a factor of $2$. Using the spatial part of \eqref{eq:04-12} and \eqref{eq:04-19}
we find
\begin{align}
 & 2 \frac{1}{2} \mpd^2 \umi{3}{q^{(\ii)}_a} \umi{3}{q^{(\ii)}_b} \e^{-\ii \, \vec{q}^{\,(\ii)}_a \cdot \vec{k}_{\{s_1\}}} \, \e^{-\ii \, \vec{q}^{\,(\ii)}_b \cdot \vec{k}_{\{s_2\}}} \notag \\
 & \fmi{3}{h} \e^{\ii \, \vec{h} \cdot \left(\vec{q}^{\,(\ii)}_a - \vec{q}^{\,(\ii)}_b \right)} P_0(h)  \notag \\
 =\, &\mpd^2 (2\pi)^6 \fmi{3}{h} \dirac\left(\vec{h} - \vec{k}_{\{s_1\}} \right) \, \dirac\left(\vec{h} + \vec{k}_{\{s_2\}} \right) \, P_0(h) \notag \\
 =\, &\mpd^2 (2\pi)^3 \dirac\left(\sum_{s=1}^n \, \vec{k}_s\right) P_0\left(\big| \vec{k}_{\{s_1\}} \big|\right) \;.
\label{eq:04-30}
\end{align}
The next diagram is the bare $\hat{C}_{\delta_a p_b}$-line in Fig.~\ref{fig:04-3}. The calculation proceeds in largely the same way as for the first diagram except that the two
labelings contained in the label grouping now give different results since the diagram is not symmetric anymore. Using \eqref{eq:04-10},\eqref{eq:04-12} and \eqref{eq:04-20} we find
the result given in \eqref{eq:04-31} for the $a_{i_1} \ldots a_{i_m} b_{j_1} \ldots b_{j_{n-m}}$ realisation and the second realisation gives the same result just with the external
label sets $\{s_1\} = \{i_1,\ldots,i_m \}$ and $\{s_2\} = \{j_1,\ldots,j_{n-m}\}$ exchanged:
\begin{align}
 &\phantom{{}={}} \frac{1}{2} 2 \, \mpd^2 \umi{3}{q^{(\ii)}_a} \umi{3}{q^{(\ii)}_b} \e^{-\ii \, \vec{q}^{\,(\ii)}_a \cdot  \vec{k}_{\{s_1\}}} \, \e^{-\ii \, \vec{q}^{\,(\ii)}_b \cdot \vec{k}_{\{s_2\}}} \notag \\
 &\phantom{{}={}} \left( (-\ii)^2 \alpha \fmi{3}{h} \e^{\ii \, \vec{h} \cdot \left(\vec{q}^{\,(\ii)}_a - \vec{q}^{\,(\ii)}_b \right)} \, P_0(h) \, \frac{\vec{h}}{h^2} \right) \cdot \vec{K}_{\{s_2\}} \displaybreak[0] \notag \\
 &=-\alpha \mpd^2 (2\pi)^6 \fmi{3}{h} \dirac\left(\vec{h} - \vec{k}_{\{s_1\}} \right) \, \dirac\left(\vec{h} + \vec{k}_{\{s_2\}} \right) \times \notag \\
 &\phantom{{}={}-\alpha \mpd^2 (2\pi)^6 \fmi{3}{h} } P_0(h) \, \frac{\vec{h} \cdot \vec{K}_{\{s_2\}}}{h^2}  \notag \\
 &=-\alpha \mpd^2 (2\pi)^3 \dirac\left(\sum_{s=1}^n \, \vec{k}_s\right) P_0\left(\big| \vec{k}_{\{s_1\}} \big|\right) 
      \frac{\vec{k}_{\{s_1\}} \cdot \vec{K}_{\{s_2\}}}{\left(\vec{k}_{\{s_1\}}\right)^2} \;.
\label{eq:04-31}
\end{align}
Observe that the argument of $P_0$ in both expressions can be exchanged between the two label sets by using the Dirac delta distribution in front.
At this point the pattern should become clear, hence the only other diagram we will compute explicitly is the first diagram of Fig.~\ref{fig:04-4} in order to show the
combined usage of \eqref{eq:04-21} and the linear term in \eqref{eq:04-27} which give the first term in parentheses in the first line of \eqref{eq:04-32}.
Due to the symmetry of the diagram both labelings give the same result.
\begin{widetext}
\begin{align}
     & 2 \frac{1}{2} \mpd^2 \umi{3}{q^{(\ii)}_a} \umi{3}{q^{(\ii)}_b} \e^{-\ii \, \vec{q}^{\,(\ii)}_a \cdot \vec{k}_{\{s_1\}}} \, \e^{-\ii \, \vec{q}^{\,(\ii)}_b \cdot \vec{k}_{\{s_2\}}} \left(-\alpha^2 \fmi{3}{h_1} \e^{\ii \, \vec{h}_1 \cdot \left(\vec{q}^{\,(\ii)}_a - \vec{q}^{\,(\ii)}_b\right)} \frac{ \vec{K}_{\{s_1\}} \cdot \vec{h}_1}{h_1^2} \, \frac{ \vec{K}_{\{s_2\}}  \cdot \vec{h}_1}{h_1^2} P_0(h_1)\right) \left(\fmi{3}{h_2} \e^{\ii \, \vec{h}_2 \cdot \left(\vec{q}^{\,(\ii)}_a - \vec{q}^{\,(\ii)}_b\right)} \, P_0(h_2) \right) \notag \\
 =\, & -\alpha^2 \mpd^2 (2\pi)^6 \fmi{3}{h_1} \fmi{3}{h_2} \dirac\left(\vec{h}_1 + \vec{h}_2 - \vec{k}_{\{s_1\}} \right) \dirac\left(\vec{h}_1 + \vec{h}_2 + \vec{k}_{\{s_2\}} \right) P_0(h_1) \, P_0(h_2) 
      \frac{ \vec{K}_{\{s_1\}} \cdot \vec{h}_1}{h_1^2} \, \frac{ \vec{K}_{\{s_2\}} \cdot \vec{h}_1}{h_1^2} \notag \\ 
 =\, & -\alpha^2 \mpd^2 (2\pi)^3 \dirac\left(\sum_{s=1}^n \, \vec{k}_s\right) \fmi{3}{h} P_0(h) \, P_0\left(\big| \vec{h} - \vec{k}_{\{s_1\}} \big|\right)
      \frac{ \vec{K}_{\{s_1\}} \cdot \vec{h}}{h^2} \, \frac{ \vec{K}_{\{s_2\}} \cdot \vec{h}}{h^2}
\label{eq:04-32}
\end{align}
\end{widetext}
In the last line we renamed $\vec{h}_1 \rightarrow \vec{h}$. The damping function is easily calculated to be
\begin{equation}
\mathfrak{D}_{(\{s_1\},\{s_2\})} = \nexp{-\frac{\sigma_p^2}{2} \left( \vec{K}_{\{s_1\}}^2 + \vec{K}_{\{s_2\}}^2 \right)} \;.
\label{eq:04-33}
\end{equation}
Once all diagrams have been calculated we gather the results into two functions as described in our strategy of section \ref{org_comb}. The general 2-particle $n$-point density
cumulant up to second order in the initial correlations then reads
\begin{align}
 G^{(0,2)}_{\rho_1 \ldots \rho_n} =& \, \mpd^2 (2\pi)^3 \dirac\left(\sum_{j=1}^n \vec{k}_j\right) \, \sum_{m=1}^{\lfloor n/2 \rfloor} \sum_{\substack{ \{s_1\},\{s_2\} \\ \#\{s_1\} = m \\ \#\{s_2\} = n-m} } \notag \\
                                                & \left( \mathcal{T}^{(2,1)}_{\#\{s_1\} | \#\{s_2\}}(\{s_1\};\{s_2\}) + \mathcal{T}^{(2,2)}_{\#\{s_1\} | \#\{s_2\}}(\{s_1\};\{s_2\}) \right) \;.
\label{eq:04-34}
\end{align}
In this, $\mathcal{T}^{(2,1)}_{\#\{s_1\} | \#\{s_2\}}$ contains all results from diagrams of first order in $P_0$ found in Fig.~\ref{fig:04-3} and $\mathcal{T}^{(2,2)}_{\#\{s_1\} | \#\{s_2\}}$ contains the results
from the second order diagrams of figure \ref{fig:04-4}. Both can be found in Appendix \ref{appA}. Their length might look daunting at first. However, they reduce
considerably when we consider explicit small $n$-point cumulants. For $n=1$ we have two terms in \eqref{eq:04-14} where in each of them one of the two particles does not
carry an external label. Theorem \ref{theo:04-3} then tells us that we have
\begin{equation}
 G^{(0,2)}_{\rho_1} = 0 \quad \text{and} \quad G^{(0,2)}_{B_1} = 0 \;.
\label{eq:04-35}
\end{equation}
In the $n=2$ case there is trivially only one possible distribution, namely $1|1$ and the only grouping is $(1;2)$. The overall Dirac delta 
allows us to set $\vec{k}_2 = -\vec{k}_1$. The $\mathcal{T}^{(2,2)}_{1|1}$ function thus reduces to
\begin{equation}
 \mathcal{T}^{(2,1)}_{1|1}(1;2) = P_0(k_1) \, (1 +  \alpha \g_1) \, (1 + \alpha \g_2) \, \nexp{-\frac{\sigma_p^2}{2} \, k_1^2 \, (\g_1^2 + \g_2^2)}
\label{eq:04-36}
\end{equation}
which is the general form of Eq.~(44) in \citet{Bartelmann2014a}. This contribution corresponds to the linear power spectrum of the standard Euler-Poisson system and is
equivalent to it if one neglects the damping and uses Zel'dovich trajectories. The $\mathcal{T}^{(2,2)}_{1|1}$ term is thus a generalization of Eq.~(48) in \citep{Bartelmann2014a}
which only contained the fourth diagram of Fig.~\ref{fig:04-4}.
\begin{gather}
 \mathcal{T}^{(2,2)}_{1|1}(1;2) =  \nexp{-\frac{\sigma_p^2}{2} \, k_1^2 (\g_1^2 + \g_2^2)} \fmi{3}{h} P_0\left(\big| \vec{k}_1 - \vec{h} \big|\right) \, \notag \\ 
 P_0(h) \left\{ \alpha^2 \g_1 \g_2 \left(\frac{\vec{k}_1 \cdot \vec{h}}{h^2}\right)^2 
 \left( 1 +\alpha \, (\g_1 + \g_2) \frac{\vec{k}_1 \cdot \left(\vec{k}_1 - \vec{h}\right)}{\left(\vec{k}_1 - \vec{h}\right)^2} \right. \right. \notag \\
   + \left. \left. \frac{\alpha^2}{2} \g_1 \g_2 \left( \frac{\vec{k}_1 \cdot \left(\vec{k}_1 - \vec{h}\right)}{\left(\vec{k}_1 - \vec{h}\right)^2} \right)^2 \right) 
   + \alpha^2 \g_1 \g_2 \left(\frac{\vec{k}_1 \cdot \vec{h}}{h^2}\right) \left( \frac{\vec{k}_1 \cdot \left(\vec{k}_1 - \vec{h}\right)}{\left(\vec{k}_1 - \vec{h}\right)^2} \right) \right\}
\label{eq:04-37}
\end{gather}
This term represents the coupling of modes at the initial time which are then transported forward by the free propagator. An interesting question for future work would be
whether comparable terms can be found in the non-interacting limit of the Euler-Poisson system or if \eqref{eq:04-37} already contains effects of multi-streaming.

With all 2-particle density-only correlators described by \eqref{eq:04-34} we can now address mixed correlators of $\rho$ and $B$. For the 2-point cumulant we directly
infer from Theorem \ref{theo:04-3} that
\begin{equation}
 G^{(0,2)}_{\rho_1 B_2} = 0 = G^{(0,2)}_{B_1 \rho_2} \;.
\label{eq:04-38}
\end{equation}
Let us thus look at the 3-point cumulant. There the only label distribution is $1|2$ with three possible label groupings so without specifying $\mathcal{T}_{1|2}^{(2)}$ we have
\begin{align}
G^{(0,2)}_{\rho_1 \rho_2 \rho_3}  =& \; \mpd^2 (2\pi)^3 \delta\left(\vec{k}_1 + \vec{k}_2 + \vec{k}_3\right) \notag \\
                                   &  \left( \mathcal{T}_{1|2}^{(2)}(1;2,3) + \mathcal{T}_{1|2}^{(2)}(2;1,3) + \mathcal{T}_{1|2}^{(2)}(3;1,2) \right) \;.
\label{eq:04-39}
\end{align}
Replacing the last density field with a $B$-field each of the $\mathcal{T}_{1|2}^{(2)}$ gets a $b(3)$ factor in front of it which depends on where $3$ is placed in the overall
grouping. This leads to
\begin{align}
G^{(0,2)}_{\rho_1 \rho_2 B_3} =&\; \ii \, \mpd^2 (2\pi)^3 \delta\left(\vec{k}_1 + \vec{k}_2 + \vec{k}_3\right) \notag \\
                                             &  \left( \vec{k}_3 \cdot \vec{k}_2 \, \g_{23} \, \mathcal{T}_{1|2}^{(2)}(1;2,3) + \vec{k}_3 \cdot \vec{k}_1 \, \g_{13} \, \mathcal{T}_{1|2}^{(2)}(2;1,3) \right) \;.
\label{eq:04-40}
\end{align}
The third term $\mathcal{T}_{1|2}^{(2)}(3;1,2)$ vanished due to Theorem \ref{theo:04-1} as one particle only carried the $B$-field label $3$. If we add one more $B$-field
this holds for all three terms, i.e.~we directly see Corollary \ref{theo:04-4} in action and thus $G^{(0,2)}_{\rho B B}=0$ . This implies a general strategy for
how to calculate any mixed correlator.
\begin{itemize}
 \item Calculate the $\mathcal{T}$ functions for the corresponding $n$-point density-only cumulant.
 
 \item Replace density by response fields as desired. For each label distribution identify then those label groupings where there are groups of only $B$-field labels and drop their
       corresponding $\mathcal{T}$-functions.
 
 \item Place $b(s)$-factors appropriate to the number of $B$-fields in front of all remaining $\mathcal{T}$-functions, where the $b(s)$ can be directly read off from the
       grouping.
\end{itemize}

\subsection{3-particle cumulants}

While technically possible, writing down general $n$-point cumulants becomes infeasible quite quickly for more than two particles. For the $\ell=3$
case we will contend ourselves with the 3-point cumulants. 4-point cumulants are given in Appendix \ref{appB}. In terms of diagrams we will only consider the lowest possible order of
$\mathcal{O}(P_0^2)$ which already has 36 diagrams. We can again reduce this number significantly by combining all non-invariant
diagrams of the same topology leading to the form of $\hat{\mathcal{C}}^{(3)}_{\mathrm{con}}$ shown in figure \ref{fig:04-5} with only 7 diagrams left.
\begin{figure}[htp]
 \centering
 \includegraphics[scale=0.5]{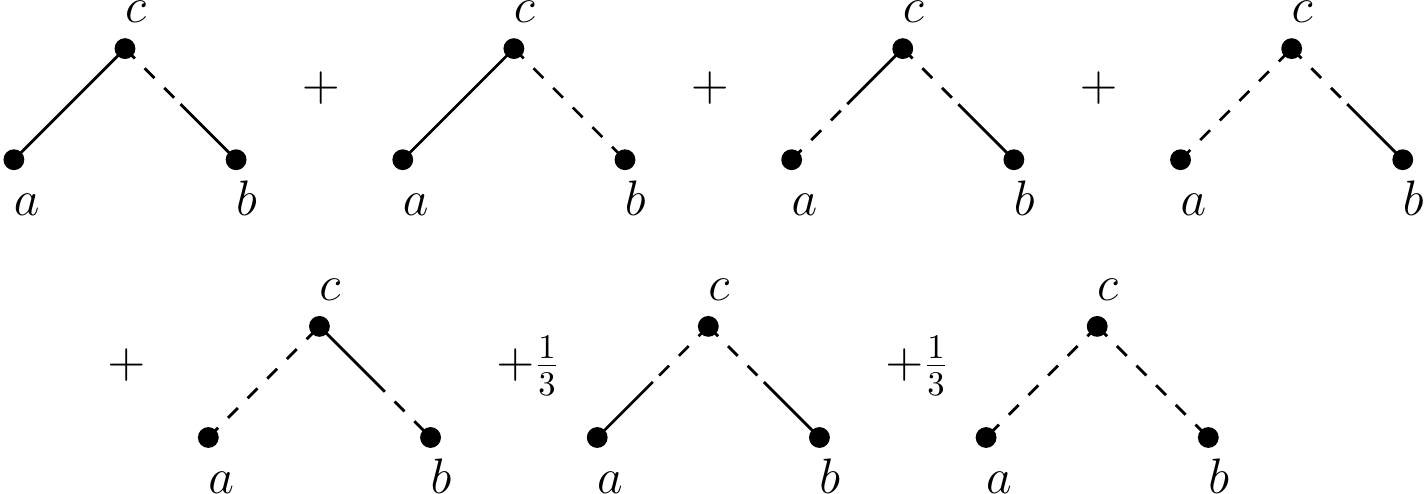}
 \caption{All diagrams of order $\mathcal{O}(P_0^2)$ contributing to $\hat{\mathcal{C}}^{(3)}_{\mathrm{con}}$ after combining those of equal topology.}
\label{fig:04-5}
\end{figure}
For the 3-point cumulant the only possible index distribution is $1|1|1$ and the only grouping subsequently $(1;2;3)$. For each of the 7 diagrams one thus needs to evaluate
$3!=6$ terms leading to a total of 42 individual contributions which make up 
\begin{align}
 \mathcal{T}_{1|1|1}^{(3,2)}(1;2;3) &= \nexp{-\frac{\sigma_p^2}{2} \, \left(\vec{K}_1^2  + \vec{K}_2^2 + \vec{K}_3^2 \right)} \notag \\
                                    &\phantom{{}={}} \Bigg\{ P_0(k_1) \, P_0(k_2) \, (1+\alpha \g_1) (1+\alpha \g_2) \notag \\
                                    &\phantom{{}={}\Bigg\{} \left[\alpha^2 \frac{\vec{K}_3 \cdot \vec{k}_1}{k_1^2} \frac{\vec{K}_3 \cdot \vec{k}_2}{k_2^2} - \vec{K}_3 \cdot \left(\frac{\vec{k}_1}{k_1^2} + \frac{\vec{k}_2}{k_2^2}\right) \right] \notag \\
                                    &\phantom{{}={}\Bigg\{} + \mathrm{cyc. \, perm.} \Bigg\} \;.
\label{eq:04-41}
\end{align}
Due to Corollary \ref{theo:04-4} the density-only cumulant is the only 3-particle 3-point cumulant.
\begin{equation}
 G^{(0,3)}_{\rho\rho\rho}(1,2,3) = \mpd^3 (2\pi)^3 \dirac\left(\vec{k}_1 + \vec{k}_2 + \vec{k}_3\right) \, \mathcal{T}_{1|1|1}^{(3,2)}(1;2;3)
\label{eq:04-42}
\end{equation}
This is again a generalization of a result from \citep{Bartelmann2014a}, namely for the bispectrum given in Eq.~(77).

\section{Summary}

This paper extends the statistical field theory for classical particles pioneered by \citeauthor{Das2012} and adapted to correlated initial conditions in
\citet{Bartelmann2014} from the canonical to the grand canonical ensemble. We were able to do this in the case of statistically homogeneous and isotropic systems by using the same
conceptual link between canonical and grand canonical ensembles as in standard statistical physics and found some generalizations of results derived in
\citep{Bartelmann2014,Bartelmann2014a}. Our main results are:
\begin{itemize}
 \item Initial correlations between the degrees of freedom of the individual particles can be brought into the form of correlation operators acting on the free generating functional
       of an ideal gas. These operators can be expressed in a simple diagrammatic representation.
       
 \item Using a variant of the Mayer cluster expansion the sum over particle numbers can be transformed into an exponential of the sum over connected $\ell$-particle generating
       functionals. Having only connected cumulants as the intrinsic building blocks of perturbation theory and having some general theorems for them
       will be advantageous in the study of the interacting grand canonical ensemble.
       
 \item The ordering of the generating functional by $\ell$ representative connected particles together with the diagrammatic representation of initial correlations
       allows a structured scheme when calculating free cumulants, reducing the combinatorical efforts necessary when compared with the canonical approach.
       With the help of this scheme we derived generalizations of the cumulants found in \citet{Bartelmann2014a}.
\end{itemize}
Building upon these, our next paper \citep{Fabis2015a} will focus on generalizing the self-consistent grand canonical perturbation theory of
\citep{Mazenko2010,Mazenko2011,Das2012,Das2013} to include initial correlations.

\begin{acknowledgments}
 We wish to thank Luca Amendola, J\"urgen Berges, Robert Lilow, Manfred Salmhofer, Celia Viermann and Christof Wetterich for insightful discussions.
 This work was generously supported in parts by the Transregional Collaborative Research Centre TR~33, ``The Dark Universe'', of the German Science Foundation as well as by
 J\"urgen Berges and the Institute for Theoretical Physics at Heidelberg University.
\end{acknowledgments}

\appendix

\section{2-particle $\mathcal{T}$-functions}\label{appA}

The two functions containing the contribution from a single label grouping $(\{s_1\};\{s_2\})$ of a distribution $\#\{s_1\} | \#\{s_2\}$ to some $n=\#\{s_1\} + \#\{s_2\}$ point
cumulant in first and second order in the initial power spectrum $P_0$ are shown \eqref{eq:A-1} and \eqref{eq:A-2}. In the first order function the first term represents the
first diagram of figure \ref{fig:04-2}, the second and third term the second diagram and the fourth term the third diagram. In the second order function the first term represents
the third diagram of figure \ref{fig:04-3}, the second term the first diagram, the third and fourth term the second diagram and the fifth term the fourth diagram.
\begin{widetext}
\begin{align}
 \mathcal{T}^{(2,1)}_{\#\{s_1\} | \#\{s_2\}}(\{s_1\};\{s_2\}) =\, & \mathfrak{D}_{(\{s_1\},\{s_2\})} \, P_0 \left( \big| \vec{k}_{\{s_1\}} \big| \right)
     \left\{ 1 - \alpha \frac{ \vec{k}_{\{s_1\}} \cdot \vec{K}_{\{s_2\}} }{ \vec{k}_{\{s_1\}}^2} - \alpha \frac{ \vec{k}_{\{s_2\}} \cdot \vec{K}_{\{s_1\}} }{ \vec{k}_{\{s_2\}}^2} 
     +\alpha^2 \frac{ \vec{K}_{\{s_1\}} \cdot \vec{k}_{\{s_1\}} }{ \vec{k}_{\{s_1\}}^2} \, \frac{ \vec{K}_{\{s_2\}}  \cdot \vec{k}_{\{s_2\}}}{ \vec{k}_{\{s_2\}}^2} \right\}
\label{eq:A-1}
\end{align}
\begin{align}
 \mathcal{T}^{(2,2)}_{\#\{s_1\} | \#\{s_2\}}(\{s_1\};\{s_2\}) =\, & \mathfrak{D}_{(\{s_1\},\{s_2\})} \, \fmi{3}{h} P_0(h) \, P_0\left( \big| \vec{k}_{\{s_1\}} - \vec{h} \big| \right)
  \left\{ \alpha^2 \frac{ \vec{K}_{\{s_1\}} \cdot \vec{h}}{h^2} \, \frac{ \vec{K}_{\{s_2\}} \cdot \left( \vec{k}_{\{s_1\}} - \vec{h} \right)}{\left( \vec{k}_{\{s_1\}} -\vec{h}\right)^2} -\alpha^2 \frac{ \vec{K}_{\{s_1\}} \cdot \vec{h}}{h^2} \, \frac{ \vec{K}_{\{s_2\}} \cdot \vec{h}}{h^2}  \right. \notag \\
   &  \left. \left(1 + \alpha \, \frac{ \vec{K}_{\{s_1\}} \cdot \left( \vec{k}_{\{s_1\}} - \vec{h} \right) }{\left( \vec{k}_{\{s_1\}} -\vec{h}\right)^2} - \alpha \, \frac{ \vec{K}_{\{s_2\}} \cdot \left( \vec{k}_{\{s_1\}} - \vec{h} \right) }{\left( \vec{k}_{\{s_1\}} - \vec{h}\right)^2} 
         - \frac{\alpha^2}{2} \frac{ \vec{K}_{\{s_1\}} \cdot \left( \vec{k}_{\{s_1\}} - \vec{h}\right) }{\left( \vec{k}_{\{s_1\}} - \vec{h}\right)^2} \, \frac{ \vec{K}_{\{s_2\}} \cdot \left( \vec{k}_{\{s_1\}} - \vec{h}\right) }{\left( \vec{k}_{\{s_1\}} - \vec{h}\right)^2} \right)  \right\} \notag \\ 
\label{eq:A-2}
\end{align}
\end{widetext}

\section{3-particle 4-point cumulants}\label{appB}
The 4-point cumulant has only one index distribution $1|1|2$ but $\binom{4}{2} = 6$ different groupings. The $\mathcal{T}$ function can in principle be read off from the
3-point case of \eqref{eq:04-41} by replacing one of the single labels by a set of 2 labels and is given in \eqref{eq:A-5}.
The 4-point pure density cumulant up to second order in initial correlations then is
\begin{gather}
 G^{(0,3)}_{\rho_1 \rho_ 2 \rho_3 \rho_4} = \; \mpd^3 (2\pi)^3 \dirac\left(\sum_{j=1}^4 \, \vec{k}_j\right) \times \displaybreak[0]\notag \\
                                                      \left[ \mathcal{T}_{1|1|2}^{(3,2)}(1;2;3,4) + \mathcal{T}_{1|1|2}^{(3,2)}(1;3;2,4) + \mathcal{T}_{1|1|2}^{(3,2)}(1;4;2,3) \right. \displaybreak[0] \notag \\
                                                      \left. + \, \mathcal{T}_{1|1|2}^{(3,2)}(2;3;1,4) +  \mathcal{T}_{1|1|2}^{(3,2)}(2;4;1,3) +  \mathcal{T}_{1|1|2}^{(3,2)}(3;4;1,2) \right] \;.
\label{eq:A-3}
\end{gather}
Notice that this is again a shot-noise like contribution due to the $\mpd^3$ scaling.
The only mixed 3-particle 4-point cumulant is
\begin{gather}
 G^{(0,3)}_{\rho_1 \rho_2 \rho_3 B_4} = \ii \, \mpd^3 (2\pi)^3 \dirac\left(\sum_{j=1}^4 \, \vec{k}_j\right) \left[ \vec{k}_4 \cdot \vec{k}_3 \, g_{34} \, \mathcal{T}_{1|1|2}^{(3,2)}(1;2;3,4) \right. \notag \\
  \left. + \, \vec{k}_4 \cdot \vec{k}_2 \, g_{24} \, \mathcal{T}_{1|1|2}^{(3,2)}(1;3;2,4) + \vec{k}_4 \cdot \vec{k}_1 \, g_{14} \, \mathcal{T}_{1|1|2}^{(3,2)}(2;3;1,4) \right] \;.
\label{eq:A-4}
\end{gather}
\begin{widetext}
\begin{gather}
 \mathcal{T}_{1|1|2}^{(3,2)}(1;2;3,4) = \, \Bigg\{ P_0(k_1) \, P_0(k_2) \, (1+\alpha \g_1) (1+\alpha \g_2) \left[\alpha^2  \frac{\left(\vec{K}_3 + \vec{K}_4\right) \cdot \vec{k}_1}{k_1^2}  \frac{\left( \vec{K}_3 + \vec{K}_4\right) \cdot \vec{k}_2}{k_2^2}   - \alpha \, \left( \vec{K}_3 + \vec{K}_4\right) \cdot \left( \frac{ \vec{k}_1}{k_1^2} + \frac{\vec{k}_2}{k_2^2} \right) \right] \notag \\
  + \left[  P_0(k_1) (1 + \alpha \g_1) P_0\left( \big| \vec{k}_3 + \vec{k}_4 \big| \right) \left( 1 + \alpha \left( \vec{K}_3 + \vec{K}_4\right) \cdot \frac{\left(\vec{k}_3 + \vec{k}_4 \right)}{\left(\vec{k}_3 + \vec{k}_4 \right)^2} \right) \left( \alpha^2 \, \frac{\vec{K}_2 \cdot \vec{k}_1}{k_1^2} \frac{\vec{K}_2 \cdot \left(\vec{k}_3 + \vec{k}_4\right)}{\left(\vec{k}_3 + \vec{k}_4 \right)^2} - \alpha \, \frac{\vec{K}_2 \cdot \vec{k}_1}{k_1^2} - \alpha \, \frac{\vec{K}_2 \cdot \left(\vec{k}_3 + \vec{k}_4 \right)}{\left(\vec{k}_3 + \vec{k}_4 \right)^2} \right) \right. \notag \\
  +  (1 \leftrightarrow 2) \Bigg] \Bigg\} \, \nexp{-\frac{\sigma_p^2}{2} \, \left(\vec{K}_1^2 + \vec{K}_2^2 + \left( \vec{K}_3 + \vec{K}_4 \right)^2 \right)}
\label{eq:A-5}
\end{gather}
\end{widetext}

\bibliography{Bibliography}
\end{document}